\documentclass[pdflatex]{sn-jnl}

\usepackage[numbers,square]{natbib}
\usepackage[noend,ruled,vlined]{algorithm2e}

\SetCommentSty{mycommfont}

\usepackage{graphicx}%
\usepackage{multirow}%
\usepackage{amsmath,amssymb,amsfonts}%
\usepackage{amsthm}%
\usepackage{mathrsfs}%
\usepackage[title]{appendix}%
\usepackage{xcolor}%
\usepackage{textcomp}%
\usepackage{manyfoot}%
\usepackage{booktabs}%
\usepackage{placeins}
\usepackage{listings}%
\usepackage{mathtools}
\usepackage{subcaption}
\usepackage{xr}
\usepackage{bm}
\usepackage{xfrac}
\usepackage{hyperref}
\usepackage{array}
\usepackage{tabularx}
\usepackage{tikz}
\usetikzlibrary{shapes.geometric}
\usepackage{booktabs} 
\usepackage[version=4]{mhchem} 
\usepackage{enumitem}
\usepackage{algorithm2e}

\theoremstyle{thmstyleone}%
\newtheorem{theorem}{Theorem}
\newtheorem{corollary}[theorem]{Corollary}%
\newtheorem{proposition}[theorem]{Proposition}%
\newtheorem{lemma}[theorem]{Lemma}%
\makeatletter
\newtheoremstyle{examplesty}
{18pt plus2pt minus1pt}
{18pt plus2pt minus1pt}
{\normalfont}
{0pt}
{\bfseries}
{.}
{.5em}
{\thmname{#1}\thmnumber{\@ifnotempty{#1}{ }{#2}}%
  \thmnote{ {\the\thm@notefont(#3)}}}
\makeatother
\theoremstyle{examplesty}%
\newtheorem{remark}{Remark}%
\theoremstyle{thmstylethree}%
\newtheorem{definition}[theorem]{Definition}%

\makeatletter
\def\moverlay{\mathpalette\mov@rlay}
\def\mov@rlay#1#2{\leavevmode\vtop{%
    \baselineskip\z@skip \lineskiplimit-\maxdimen
    \ialign{\hfil$\m@th#1##$\hfil\cr#2\crcr}}}
\newcommand{\charfusion}[3][\mathord]{
  #1{\ifx#1\mathop\vphantom{#2}\fi
    \mathpalette\mov@rlay{#2\cr#3}
  }
  \ifx#1\mathop\expandafter\displaylimits\fi}
\DeclareRobustCommand\bigop[1]{%
  \mathop{\vphantom{\sum}\mathpalette\bigop@{#1}}\slimits@
}
\newcommand{\bigop@}[2]{%
  \vcenter{%
    \sbox\z@{$#1\sum$}%
    \hbox{\resizebox{\ifx#1\displaystyle.9\fi\dimexpr\ht\z@+\dp\z@}{!}{$\m@th#2$
}}%
  }%
}
\makeatother

\newcommand{\king}{K}

\newcommand{\food}{\ensuremath{\mathbf{F}}}
\newcommand{\SM}{S}

\newcommand{\closure}[2]{\mathscr{C}_{#1}(#2)}

\newcommand{\CR}[0]{\mathcal{R}}

\newcommand{\FR}[0]{\mathbf{R}}

\newcommand{\fr}[0]{\mathbf{r}}

\newcommand{\RAF}[0]{\mathbf{R}'}
\newcommand{\block}[0]{\mathbf{B}}

\newcommand{\cat}{\kappa_F}
\newcommand{\sgen}{g}
\newcommand{\rgen}{\gamma}

\newcommand{\feducts}{\rho_F}

\newcommand{\fproducts}{\pi_F}

\newcommand{\ORCID}[1]{}

\definecolor{middlegreen}{RGB}{0,128,0}
\definecolor{royalpurple}{RGB}{85,26,139}

\begin{document}

\title{Minimality in Reflexive and Stoichiometric Autocatalysis} 

\abstract{
	Autocatalysis, the ability of a chemical subsystem to sustain its own constituents when supplied with sufficient food molecules, 
	has been closely related to the origin of life on Earth. Emerging from Wilhelm Ostwald's considerations about an explicit autocatalytic 
	reaction, different notions of autocatalysis have been developed over the years. The two most prominent are \textsc{reflexively autocatalytic 
	F-generated sets (RAFs)} and \emph{stoichiometric autocatalysis}. After having shown that each \textsc{RAFs} is, under reasonable conditions, 
	in general \emph{stoichiometrically autocatalytic}, we examine here the relationship between the two notions of minimality: 
	\textsc{irreducible RAFs} and \emph{autocatalytic cores}. To this end, we overcome the obstacle that \textsc{RAFs} and 
	\emph{stoichiometric autocatalysis} have been formalized in distinct systems of chemical reactions, i.e., catalytic reaction 
	systems (CRS) and chemical reaction networks (CRNs), respectively. We show that reactions in a CRS constitute equivalence 
	classes of reactions of the corresponding CRNs w.r.t. their specific catalyzations. Using the fact that CRN and CRS can be canonically identified 
	whenever each CRS reaction is associated with a single catalyzation, we demonstrate that the K\H{o}nig graph of a monocatalyzed, 
	irreducible RAF is composed of strong blocks devoid of food and waste species that are separated by reaction vertices, each of which contains 
	an autocatalytic core. In fact, a single irreducible RAF can, in general, contain exponentially many autocatalytic cores.
}

\author*[1,4]{\fnm{Richard} \sur{Golnik}}\email{richard@bioinf.uni-leipzig.de}
\ORCID{0000-0002-8582-5006} 

\author[1]{\fnm{Thomas} \sur{Gatter}}\email{thomas@bioinf.uni-leipzig.de}
\ORCID{0000-0002-5016-5191} 

\author[1]{\fnm{Wim} \sur{Hordijk}}\email{wim@worldwidewanderings.net}
\ORCID{0000-0002-0223-6194} 

\author[1,2,3,4,5,6,7,8,9]{\fnm{Peter F.}\sur{Stadler}}\email{studla@bioinf.uni-leipzig.de}
\ORCID{0000-0001-5567-3016}

\author[1,2]{\fnm{Nicola} \sur{Vassena}}\email{nicola.vassena@uni-leipzig.de}
\ORCID{0000-0001-5411-4976} 

\affil[1]{\orgdiv{Bioinformatics Group, Department of Computer Science}
  \orgname{Leipzig University}, \orgaddress{\street{H{\"a}rtelstra{\ss}e 16–18},
    \postcode{D-04107} \city{Leipzig}, \country{Germany}}} 

\affil[2]{\orgdiv{Interdisciplinary Center for Bioinformatics}
	\orgaddress{\orgname{Leipzig University}, \postcode{D-04107}
	 \city{Leipzig}, \country{Germany}}}  
	 
\affil[3]{\orgdiv{Center for Scalabale Data Analytics and Artificial Intelligence},
  \orgaddress{\orgname{Leipzig University}, \postcode{D-04107}
    \city{Leipzig}, \country{Germany}}} 

\affil[4]{\orgdiv{Zuse School for Embedded and Composite Artificial Intelligence (SECAI)}
   \orgaddress{\orgname{Leipzig University}, \postcode{D-04107}
    \city{Leipzig}, \country{Germany}}}

\affil[5]{\orgname{Max Planck Institute for Mathematics in the Sciences},
  \orgaddress{\street{Inselstra{\ss}e 22}, \postcode{D-04103}
    \city{Leipzig}, \country{Germany}}}

\affil[6]{\orgdiv{Department of Theoretical Chemistry}, \orgname{University
    of Vienna}, \orgaddress{\street{W{\"a}hringerstra{\ss}e 17},
    \postcode{A-1090} \city{Wien}, \country{Austria}}}

\affil[7]{\orgdiv{Facultad de Ciencias}, \orgname{Universidad Nacional de
    Colombia}, \orgaddress{\city{Bogot{\'a}}, \country{Colombia}}}

\affil[8]{\orgdiv{Center for non-coding RNA in Technology and Health},
  \orgname{University of Copenhagen}, \orgaddress{\street{Ridebanevej 9},
    \postcode{DK-1870} \city{Frederiksberg}, \country{Denmark}}}

\affil[9]{\orgname{Santa Fe Institute}, \orgaddress{\street{1399 Hyde Park
      Rd.}, \city{Santa Fe}, \state{NM} \postcode{87501}, \country{USA}}}



\maketitle

\section{Introduction}

Wilhelm Ostwald's reflections in 1890 on whether chemical reactions
could be catalyzed by one of their own products \cite{Ostwald:90} marked
the provenance of autocatalysis. In its simplest form, autocatalysis
appears as an autocatalytic reaction where one participating species drives
its own formation, fueled by a sufficient supply of food molecules. In
subsequent decades, \emph{autocatalytic reactions} appeared, either
explicitly or implicitly, in the literature across a wide range of
reaction-network settings, including chemistry, ecology, population
dynamics, and epidemiology \cite{lotka_contribution_1910,
  lotka_analytical_1920, lotka1925elements, volterra1926variazioni,
  kermack_contributions_1991}. One of the most convincing examples is the
transmission of viral infections from one individual to another
\cite{kermack_contributions_1991}, potentially resulting in an exponential
explosion of numbers of infected individuals, e.g., as observed during the
COVID-19 pandemic \cite{saxena_propagation_2021}.

From the simple idea of an autocatalytic reaction, a growing field of
research has developed over the past decades, resulting in different
concepts of autocatalysis \cite{andersen_defining_2021}.  A first line of
studies focused on entities driving their own reproduction in complex
networks, so-called \emph{replicators}, such as Eigen's quasi-species
\cite{eigen_selforganization_1971}, an ensemble of genotypes that arrange
themselves around a common species from which they differ by a series of
mutations.  Subsequent research applied this concept to metabolic networks
to identify important metabolic \emph{autocatalytic replicators} such as
ATP \cite{kun_computational_2008}. More generally, sequences of reactions
whose combined net reaction is autocatalytic have been termed autocatalytic
cycles, and shown to be prevalent in the central carbon metabolism of
\emph{Escherichia coli} \cite{barenholz_design_2017,
  smith_universality_2004}. Similar ideas have subsequently been studied
\cite{joshi_autocatalytic_2021}, as they closely relate to the hypercycle
concept originally proposed by Eigen and Schuster
\cite{eigen_principle_1977}.

A concise definition of \emph{stoichiometric autocatalysis} was proposed in
\cite{blokhuis_Universal_2020} as a generalization of the idea of an
autocatalytic cycle in accordance with the IUPAC definition of
autocatalytic reaction \cite{gold_iupac_2025}.  This formalization has
received increasing attention since its conception in 2020. It has lead to
the development of algorithms for enumerating minimal autocatalytic
subsystems in fully reversible \cite{kosc_thermodynamic_2025} CRNs and
with reversible and irreversible reactions \cite{gagrani_polyhedral_2024,
  golnik_birne_2025, golnik_enumeration_2026,
  golnik_using_2026}. Originally introduced in CRNs
lacking explicitly catalyzed reactions, the theory on stoichiometric
autocatalysis has recently been extended to CRNs allowing explicit catalysis as well
\cite{golnik_autocatalytic_2026}.

An alternative approach to formalize system-level autocatalysis focused on
systems of chemical reactions in which every reaction is, in general,
explicitly catalyzed \cite{lohn_evolving_1998, steel_emergence_2000}. The
original idea dates back to Kauffman's \cite{kauffman_autocatalytic_1986}
ideas about how life on Earth might have emerged as a metabolism-first 
event rather than a genetics-first one. Kauffman explored
whether sets of polypeptide complexes, generated via reactions that are
exclusively catalyzed by members of the set, could lead to collectively
self-reproducing systems.  Remarkably, under reasonable assumptions and a
sufficient supply of amino acid monomers, their existence becomes almost
guaranteed as long as the polypeptide chains exceed a certain length. These
ideas laid the foundation for later work on \textsc{Reflexively
  Autocatalytic Food-generated sets} (RAF), which were studied extensively
in many variations \cite{steel_emergence_2000, mossel_random_2005,
  steel_minimal_2013, smith_autocatalytic_2014, hordijk_history_2019,
  steel20}.

Related to different approaches to thinking about the problem of how life
has emerged on a primordial Earth \cite{kauffman_autocatalytic_1986,
  peng_hierarchical_2022, unterberger_stoechiometric_2022},
\emph{stoichiometric autocatalysis} and \textsc{RAF sets} differ
substantially in their motivation, setting, and mathematical
formalization. Notably, \emph{stoichiometric autocatalysis} is grounded in
the tradition of chemical reaction network theory, tracing back to
Rutherford Aris \cite{aris_prolegomena_1965}, while \textsc{RAF} theory is
phrased in terms of catalytic reaction systems (CRS) \textit{sensu} Lohn
\cite{lohn_evolving_1998}. In particular, the RAF framework allows
reactions to also occur uncatalyzed, at least initially, which is
reasonable from a chemical perspective and considering the time span life
required to develop. In a more restrictive setting, i.e., in
\textsc{Constructively Autocatalytic F-generated sets} (CAFs), reactions
can only be executed if catalysts have been generated beforehand. A
meaningful connection between stoichiometric autocatalysis and RAF theory
has been investigated only very recently \cite{golnik_bridging_2026}, where
it is shown that any RAF that is itself not entirely a CAF is always
\emph{stoichiometrically autocatalytic}.


Both the CRN and the CRS frameworks have developed parallel concepts
of `minimal’ autocatalytic sub-systems. Although they differ
substantially, they share the notion of minimality by considering
autocatalytic sub-systems that do not themselves contain even smaller
autocatalytic sub-systems. The resulting minimal structures are termed
\emph{autocatalytic cores} in stoichiometric autocatalysis and
\emph{irreducible} RAFs in RAF theory. Computational tools have been
developed to identify minimal structures in both frameworks 
\cite{kosc_thermodynamic_2025, gagrani_polyhedral_2024, 
golnik_using_2026, hordijk_algorithms_2026}.
Developed only recently, \texttt{autogatito} \cite{golnik_using_2026} can
enumerate autocatalytic cores in large chemical reaction networks, for
example, of metabolic type. A na{\"\i}ve and informal application of
\texttt{autogatito} to the chemical reaction networks associated with
irreducible RAFs generated from random instances of the Binary Polymer
Model \cite{kauffman_autocatalytic_1986} produced an unexpected
outcome: each irreducible RAF contains a large number of autocatalytic 
cores -- up to $5000$ for a single irreducible RAF with approximately 50 
reactions (see Fig.~\ref{fig:irrRAFCores}). The need to
explain this apparent discrepancy motivates the present contribution, in
which we investigate the structural properties of irreducible RAFs.

\begin{figure}[htb]
  \centering
  \begin{minipage}[c]{\textwidth}
    \centering
    \begin{minipage}[c]{0.5\textwidth}
      \centering
      \includegraphics[width=\textwidth]{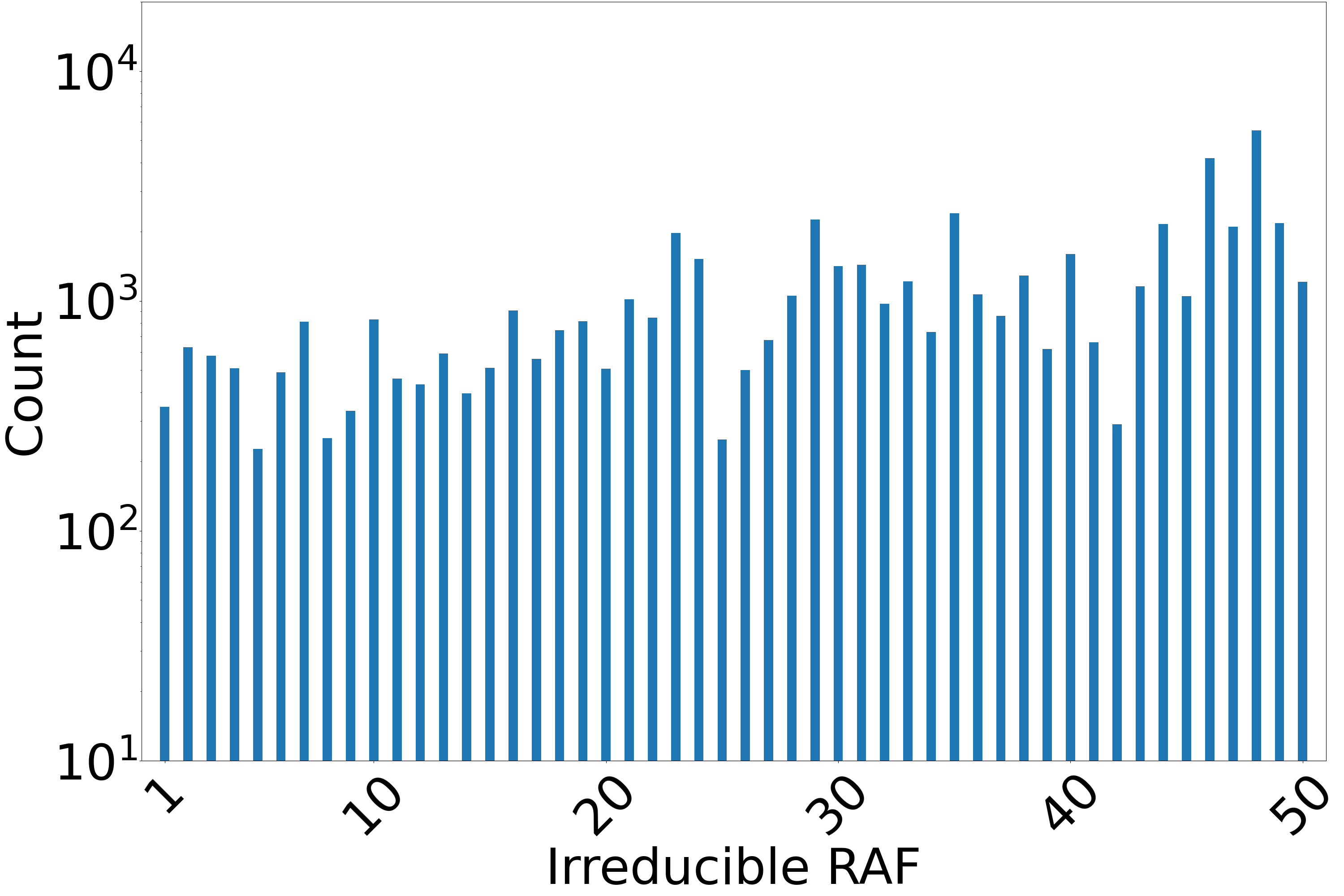}
    \end{minipage}%
    \hfill%
    \begin{minipage}[c]{0.5\textwidth}
      \centering
      \includegraphics[width=\textwidth]{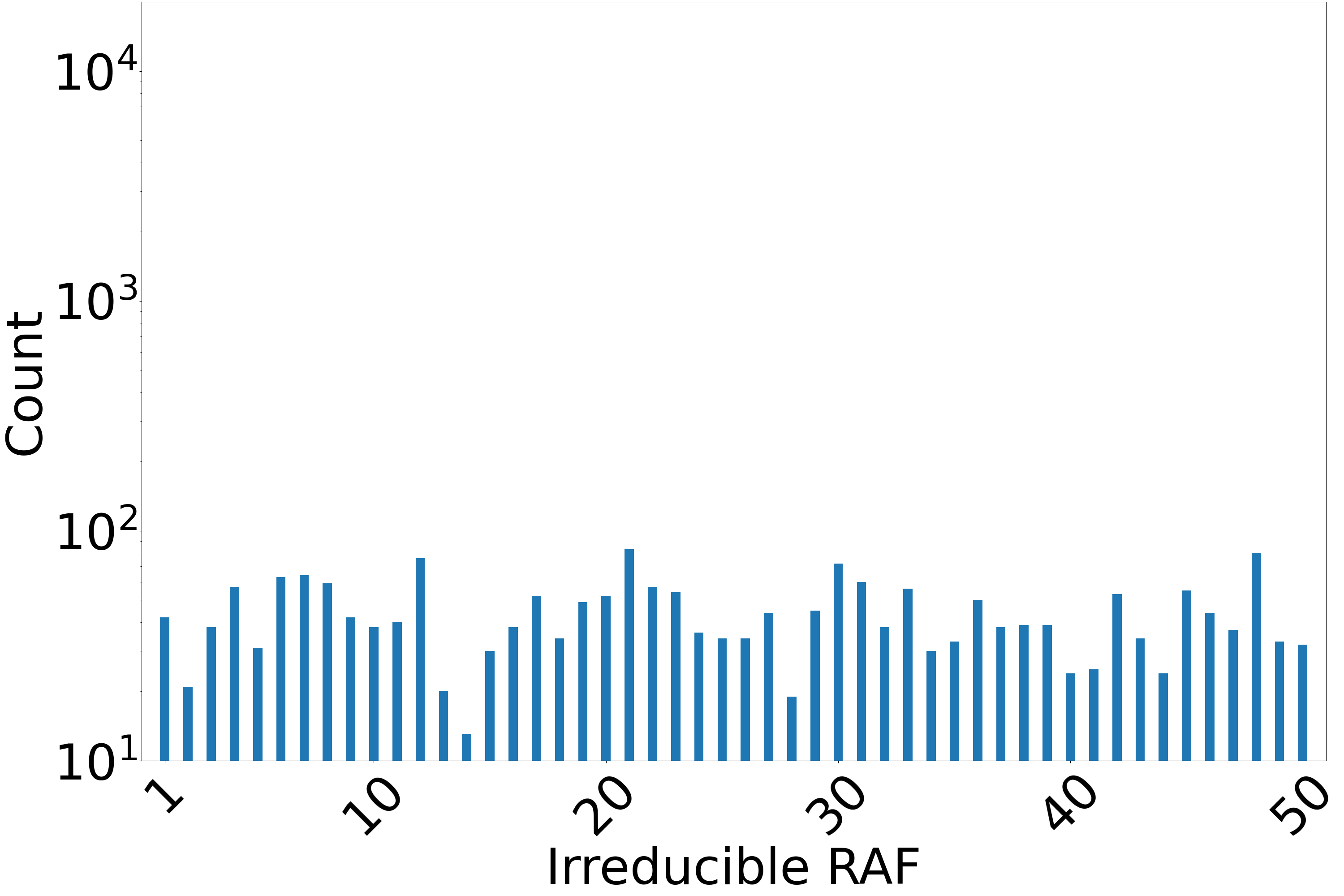}
    \end{minipage}
  \end{minipage}
  \begin{minipage}[c]{\textwidth}
    \centering
    \begin{minipage}[c]{0.5\textwidth}
      \centering
      \includegraphics[width=\textwidth]{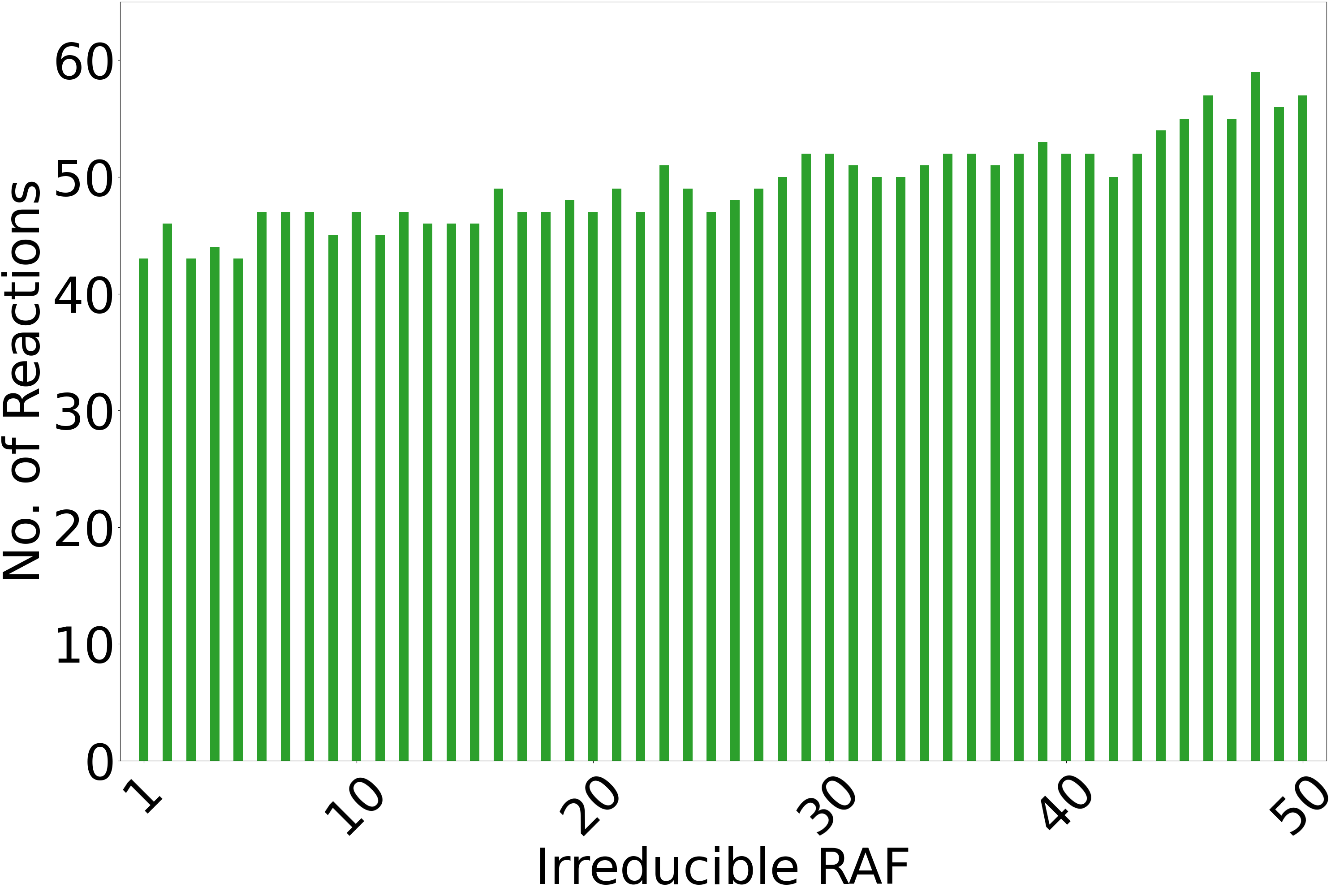}
    \end{minipage}%
    \hfill%
    \begin{minipage}[c]{0.5\textwidth}
      \centering
      \includegraphics[width=\textwidth]{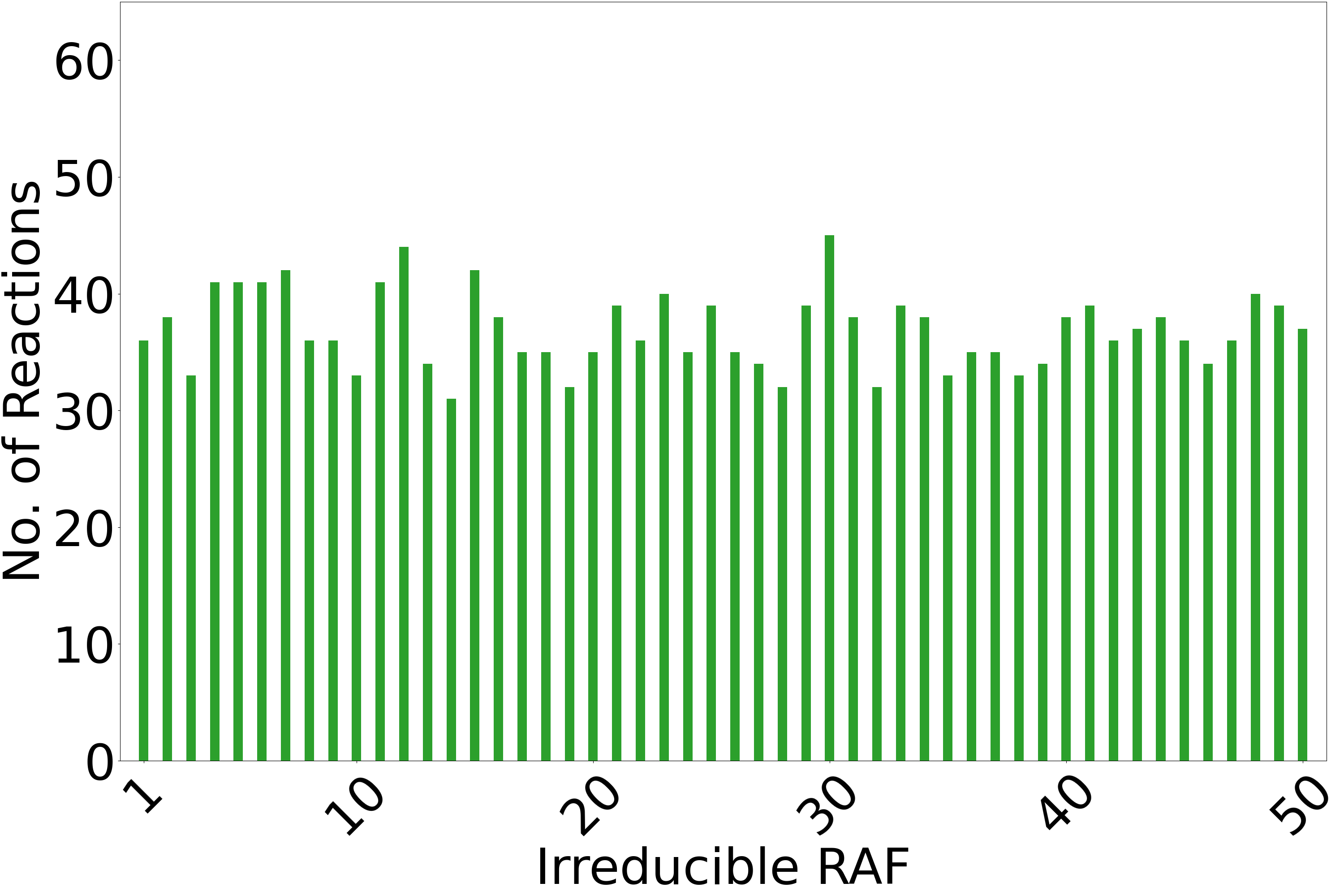}
    \end{minipage}
  \end{minipage}	
  \caption{A single irreducible RAFs may contain many autocatalytic
    cores. The plots show the number of autocatalytic cores (\textbf{top})
    and reactions (\textbf{bottom}) in 50 arbitrarily chosen irreducible
    RAFs from two instances of the Binary Polymer Model (BPM) with polymers
    of up to six monomers. We use the term ‘reaction’ here informally; the
    distinction between reactions in a CRN and `F-reactions' in a CRS is
    clarified in Sect.~\ref{sect:term}. The \textbf{left} panels reflect a
    BPM including both ligation and cleavage reactions with a fixed
    probability for a non-food species to catalyze any reaction of
    $p=0.0025$. The \textbf{right} panels reflect a model containing
    ligation reactions only, with a fixed probability for a non-food
    species to catalyze any reaction of $p=0.006$. See SI
    Sect.~\ref{sec:BinPol} as well as SI
    Table~\ref{tab:SimulationDetails} for simulation details, and SI
    Fig.~\ref{fig:irrRAFCoreComp} for composition of each irreducible
    RAF.}
  \label{fig:irrRAFCores}
\end{figure}

We first clarify in Sec.~\ref{sec:CRNCRS} the formal relationship between
the two underlying frameworks, chemical reaction networks (CRNs) and
catalytic reaction systems (CRSs), so that we may move freely between them:
Intuitively, CRSs can be viewed as quotient structures of associated
CRNs. We then study the structure of irreducible RAFs, particularly from a
graph-theoretic perspective, in Sec.~\ref{sec:RAF}. Through the quotient
relation between CRN and CRS, to each irreducible RAF we will associate an
exponentially large family of certain bipartite König
graphs. Sec.~\ref{sec:AutoCatIrrRAF} introduces stoichiometric
autocatalysis and our main result, Thm.~\ref{Thm:StrBlockCore}, which
states that every maximal \emph{strong block} \cite{grotschel_minimal_1979}
of each such bipartite graph contains an autocatalytic core. 
 This exponential correspondence provides an intuition for the large number of
autocatalytic cores observed in the irreducible RAFs of
Fig.\ref{fig:irrRAFCores}. We comment extensively on this aspect in
Sec.~\ref{sec:examples}. Nevertheless, we also show that this relationship
depends strongly on the type of network: we present an artificial
construction, which can produce examples containing any arbitrary number of
irreducible RAFs, yet invariably only one single autocatalytic core. We
conclude with a discussion of future directions.
 
\FloatBarrier

\section{Two Representations of Chemical Systems}\label{sec:CRNCRS}

Research on systems of chemical interactions has attracted considerable
attention, but formalizations lack consistency throughout the various
frameworks \cite{lohn_evolving_1998, steel_emergence_2000,
  feinberg_foundations_2019}. In particular, different authors have
introduced different notions of reactions \cite{muller_what_2022}. Here, we 
overview two approaches: \emph{Chemical Reaction Networks} and
\emph{Catalytic Reaction Systems}.

\subsection{Chemical Reaction Networks (CRNs)}
\label{sect:term}
A \emph{chemical reaction network} (CRN) $\Gamma \coloneqq (X,\CR)$ is a
tuple of two finite sets, a set of species $X$ and a set of reactions
$\CR$.  Any $r \in \CR$ is represented as an ordered association between a
nonnegative combination of reactant species and a nonnegative combination of
product species,
\begin{equation}\label{eq:reaction}
  \sum_{x\in X}\ s^-_{xr} \cdot x \quad
  \underset{r}{\longrightarrow} \quad
  \sum_{x\in X} \ s^+_{xr}\cdot x\,,
\end{equation}
where $s^{\pm}_{xr}\ge0$ are the reactant and product stoichiometry of the
molecule $x$ in the reaction $r$. The $|X|\times|\CR|$ \emph{stoichiometric
  matrix} $\SM$ encodes all \emph{net-stoichiometries}:
\begin{equation} 
  \SM_{xr} \coloneqq s^+_{xr}-s^-_{xr}.
\end{equation}

A \emph{sub-CRN} $\Gamma'\subseteq \Gamma$ is a tuple $\Gamma'=(X',\CR')$
with $X'\subseteq X$ and $\CR'\subseteq \CR$, where $r\in \CR'$ omits the
reference to species $x\not\in X'$, and the associated stoichiometric
matrix is denoted as $S[X',\mathcal{R}']$.

Clearly, stoichiometry and net-stoichiometry might not coincide, i.e.,
$s^{\pm}_{xr}\neq \pm S^{\pm}_{xr}$ for example due to
catalysis. From a chemical perspective, a \emph{catalyst} for $r$ is a
species $x\in X$ that appears on both sides of the stoichiometric equation
\eqref{eq:reaction}, i.e. with $s^-_{xr}s^+_{xr}\neq 0$, and that
additionally ``increases the rate of a reaction without modifying the
overall standard Gibbs energy change in the reaction''
\cite{gold_iupac_2025}.

However, a \emph{solvent} $x$ in a reaction $r$ may also formally satisfy
$s^-_{xr}s^+_{xr}\neq 0$ without being considered a catalyst
\cite{marx_nature_1999}, e.g. water in Br{\o}nstedt/Lowry's acid-base
theory \cite{eigen_proton_1964}. To be specific on the catalytic role a
species plays, we partition stoichiometries as in
\cite{golnik_bridging_2026} as
\begin{equation}
  s^{\pm}_{xr}=\sigma^{\pm}_{xr} + k_{xr},\qquad \sigma^{\pm}_{xr}, k_{xr}\ge0.
\end{equation}
It is readily verified that
$\SM_{xr} = s^+_{xr}-s^-_{xr} = \sigma^+_{xr} - \sigma^-_{xr}$, thus the
stoichiometric matrix is identically encoded in the coefficients
$\sigma^{\pm}_{xr}$.  Eq.~\ref{eq:reaction} can then be rewritten as:
\begin{equation}\label{eq:Mod.SmEq}
  \sum_{x\in X}\ \sigma^-_{xr} \cdot x + \left(\sum_{x\in X}\ k_{xr}
    \cdot x\right) \quad
  \underset{r}{\longrightarrow} \quad
  \sum_{x\in X} \ \sigma^+_{xr} \cdot x +
  \left(\sum_{x\in X}\ k_{xr} \cdot x\right),
\end{equation}
We refer to species appearing in Eq.~\ref{eq:Mod.SmEq} with
$\sigma^{-}_{xr}>0$, $\sigma^{+}_{xr}>0$, and $k_{xr}>0$ as F-reactants,
F-products, and F-catalysts, respectively. The nomenclature ``F-catalysts''
is to distinguish this notion from the more loose definition of a catalyst
as a species that is both reactant and product, which is often also found
in CRN literature. We collect F-reactants, F-products, and F-catalysts in
the following sets:
\begin{equation}\label{eq:setsdef}
  \begin{split}
    \feducts(r) & \coloneqq \{x\in X \ | \ \sigma^-_{xr} > 0\}; \\
    \fproducts(r) & \coloneqq \{x\in X \ | \ \sigma^+_{xr} > 0\}; \\
    \cat(r)&\coloneqq \{(x, k_{xr})\ | \ x\in X,\; k_{xr}>0\}.
  \end{split}
\end{equation} 
For technical reasons (to be clarified below in Sec.~\ref{sec:CRS}), note
that the set of F-catalysts is defined as a set of pairs $(x,k_{xr})$ to
encode also the respective coefficient. For simplicity, with an abuse of
notation, we will write $x\in \cat(r)$ to refer to species $x$ for which
there exists a pair $(x,k_{xr})\in \cat(r)$. In a CRN, distinct reactions
$r_1\neq r_2$ may differ only in the F-catalysts
$\cat(r_1)\neq \cat(r_2)$ and share identical F-reactants and
F-products. For example, the hydrogenation of ethene to ethane can be
mediated by platinum, palladium, or nickel separately
\cite{crampton_ethylene_2016}. An example for co-catalysis, i.e., where the
orchestrated cooperation of two catalysts is required, is the dual-catalyst
system of palladium and copper in the Hoechst-Wacker process, where ethene
is oxidized to acetaldehyde \cite{jira_acetaldehyde_2009}. The idea that
one `same' chemical process, intended as a net-molecule transformation, can
be catalyzed by different catalysts has led to the different, although
related, concept of Catalytic Reaction Systems, which we overview in the
next section.

\subsection{Catalytic Reaction Systems (CRS)}\label{sec:CRS}

\emph{Catalytic Reaction Systems} (CRS) \textit{sensu} 
Lohn \cite{lohn_evolving_1998} are defined as a triple 
\begin{equation}\mathbf{\Phi}\coloneqq (X, \FR, \cat(\FR)),
\end{equation}
where $X$ is again a set of species, and each $\fr\in\FR$ is defined as a
net-molecule transformation
\begin{equation}\label{eq:FReacEq}
  \sum_{x\in X}\ \sigma^-_{x\fr} \cdot x \quad
  \underset{\fr}\longrightarrow \quad
  \sum_{x\in X}\ \sigma^+_{x\fr} \cdot x,
\end{equation}
where $\sigma^\pm_{x\fr}\ge 0$ are again nonnegative stoichiometric
coefficients. In the original framework of Lohn, general stoichiometric 
coefficients are not specified, as reactions are defined through polymer 
transformations with unit stoichiometry $\sigma^\pm_{x\fr}\in \{0,1\}$. 
Here, to bring the CRS formalism closer to the general CRN setting, 
we allow arbitrary stoichiometry $\sigma^\pm_{x\fr}>0$. 
The sets of F-reactants $\feducts(\fr)$ and F-products
$\fproducts(\fr)$ of an F-reaction $\fr$ are defined identically as in
\eqref{eq:setsdef} for CRNs. The main difference with the CRN setting is
the \emph{catalyzations} $\cat(\FR)=\{ \cat(\fr) \}_{\fr\in\FR}$.  For a
single F-reaction $\fr$, the set of catalyzations $\cat(\fr)$ is defined as
$\cat(\fr)=\{c_{\fr,i}\}$, where each catalyzation $c_{\fr,i}$ is in turn a set of
pairs $(x,k^{c}_{x\mathbf{r}})$:
\begin{equation}\label{eq:CatCFR}
    c_{\fr,i}\coloneqq \{(x_1,k^c_{x_1\mathbf{r}}),\ldots, (x_n,k^c_{x_n\fr})\}.
\end{equation}
where $k^{c}_{x\fr}>0$ is a positive coefficient. Again, for simplicity and
with abuse of notation, we will write $x\in c_\fr$ for a species $x$ such
that there exists $k^c_{x\fr}$ for which $(x,k^c_{x\fr})\in c_{\fr}$.  We
call an $\fr\in \FR$ \emph{F-reaction} to distinguish it from the CRN
reactions \eqref{eq:Mod.SmEq}. A CRS
$\mathbf{\Phi}'\coloneqq (X', \FR', \cat'(\FR'))$ is a \emph{sub-CRS} of
$\mathbf{\Phi}\coloneqq (X, \FR, \cat(\FR))$ if \[X'\subseteq X,\quad
\FR'\subseteq\FR,\quad \cat'(\FR')\subseteq \cat(\FR).\]

Finally, a comparison between a CRN reaction \eqref{eq:Mod.SmEq} with a CRS
F-reaction \eqref{eq:FReacEq} and the definition of catalyzations set
$\cat$ suggests the following formalization of the relation between a CRN
and a CRS.

\paragraph{From a CRN $\Gamma=(X,\mathcal{R})$ to a CRS
  $\mathbf{\Phi}\coloneqq (X, \FR, \cat(\FR))$.}
Consider the equivalence relation $\sim$ on CRN reactions defined as
\begin{equation}\label{eq:EqRel}
    r_1 \sim r_2 \quad \text{if $\sigma_{xr_1}^{\pm}=\sigma_{xr_2}^{\pm}$ for any $x\in X$.}
\end{equation}
To each CRN $\Gamma=(X,\mathcal{R})$ we associate a CRS
$\mathbf{\Phi}(\Gamma)\coloneqq (X, \FR, \cat(\FR))$ where $\FR$ is just
the quotient set of $\mathcal{R}$ over the equivalence relation $\sim$,
i.e.
\begin{equation}
    \FR= \mathcal{R}/\sim.
\end{equation}
In particular, an F-reaction in a CRS is just an equivalence class of CRN
reactions, $\fr=[r]_{\sim}$. Each catalyzation $c$ is then just the set of
F-catalysts $c(\mathbf{r})=\cup_ {r\in \fr}\{\cat(r)\}$.

\paragraph{From a CRS $\mathbf{\Phi}\coloneqq (X, \FR, \cat(\FR))$ to a CRN
  $\Gamma=(X,\mathcal{R})$.} Rolling back the previous
argument, to each CRS $\mathbf{\Phi}\coloneqq (X, \FR, \cat(\FR))$ we
associate a CRN $\Gamma (\mathbf{\Phi})=(X,\mathcal{R})$, where each
reaction $r\in \mathcal{R}$ corresponds to a pair $(\fr,c_\fr)$, naturally
defined from \eqref{eq:FReacEq}, setting $\cat(r)=c_\fr$, and defining $r$
via \eqref{eq:Mod.SmEq} with $\sigma^{\pm}_{xr}=\sigma^{\pm}_{x\fr}$.

By construction, it is clear that 
\begin{equation}
  \Gamma(\mathbf{\Phi}(\Gamma))=\Gamma \quad \text{and}
  \quad \mathbf{\Phi}(\Gamma(\mathbf{\Phi))}=\mathbf{\Phi}.
\end{equation}
Moreover, since $\FR$ is a quotient set from $\mathcal{R}$, we have that
$|\FR| \le |\mathcal{R}|$, with equality achieved if, and only if, each
F-reaction $\fr$, viewed as an equivalence on $\mathcal{R}$, consists only
of one element, i.e.  $\fr=[r]_{\sim}=\{r\}$.

\section{Reflexively Autocatalytic Food-Generated Sets}\label{sec:RAF}

Catalytic Reaction Systems $(X,\FR,\cat(\FR))$ are the framework on which
the notion of \emph{Reflexively Autocatalytic F-generated Sets} (RAFs) has
been introduced \cite{steel_emergence_2000, hordijk_detecting_2004}. Their
definition relies on the concept of \emph{closure}. The closure
$\closure{\RAF}{X'}$ of a species set $X' \subseteq X$ w.r.t.\ to a
F-reaction set $\RAF \subseteq \FR$ is the unique inclusion-minimal subset
$Y\subseteq X$ containing $X'$ such that for any
$\fr\in \RAF: \feducts(\mathbf{r})\subseteq Y \Rightarrow
\fproducts(\mathbf{r})\subseteq Y$. In symbols:
\begin{equation}\label{eq:Closure}
  \closure{\RAF}{X'} \coloneqq \min_{\subseteq}
  \{Y  \ | \ X'\subseteq Y \subseteq X,\;
  \feducts(\mathbf{r}) \subseteq Y \Rightarrow \fproducts(\mathbf{r})
  \subseteq Y \ \forall
  \mathbf{r}\in \RAF\}.
\end{equation}
A Reflexively Autocatalytic F-generated set (RAF) is then defined as
follows \cite{hordijk_required_2011}:

\vspace{6pt}
\begin{definition}[RAF]\label{def:RAF}
  Let $(X,\FR,\cat(\FR))$ be a CRS with a specified \emph{food set}
  $\food\subseteq X$. A set $\RAF\subseteq \FR$ is called \vspace{6pt}
  \begin{itemize} 
  \item[i.]{\emph{Reflexively Autocatalytic} (RA) if, for each
      $\mathbf{r}\in \RAF$,
      there exists at least one catalyzation $c\in \cat(\mathbf{r})$
      such that $c \subseteq \closure{\RAF}{\food}$;}\\[-1.0em]
  \item[ii.]{\emph{\food-generated} (F) if $\feducts(\RAF)\subseteq
      \closure{\RAF}{\food}$;}\\[-1.0em]
  \item[iii.]{\emph{Reflexively Autocatalytic and $\food$-generated}
      (RAF) if both i. and ii. hold.}
  \end{itemize}
\end{definition}
\begin{remark}\label{rmk:RAFCRN}
  Def.~\ref{def:RAF} does not formally require the formalism of CRS, and it
  can be stated identically for CRN so that we may naturally
  speak of sets of food-generated CRN reactions, as well. This route was
  undertaken e.g. in \cite{golnik_bridging_2026}.
\end{remark}
We also remark that Property~$ii.$ together with the definition of the
closure, Eq.~\eqref{eq:Closure}, directly implies that all F-products of
any F-reaction $\fr\in \RAF$ are also contained in $\closure{\RAF}{\food}$,
thus $\fproducts(\fr)\subseteq \closure{\RAF}{\food}$. Finally, it is also
worthwhile noting that Property~$i.$ and $ii.$ are independent of each
other. In particular, Property~$ii.$ depends only on $\food$ and $\RAF$,
and not on the catalyzations. This reveals an inherent freedom in the
definition of catalyzations of a RAF, as the following Lemma underlines.

\vspace{6pt}
\begin{lemma}\label{lem:ArbCat}
  Let $X$ be a set of species and $\FR$ be a set of $F$-reactions on
  $X$. Then, for any $\food$-generated set $\RAF$ there exist infinitely
  many sets of catalyzations $\{\cat(\FR)_i\}$ which make $\RAF$ a RAF in
  the CRS $(X,\FR,\cat (\FR)_i))$, for any $i$.
\end{lemma}
\begin{proof}
  To define such valid set of catalyzations $\cat(\FR)$, it is sufficient
  to arbitrarily assign to each F-reaction $\mathbf{r}$ any set
  $\{c_j(\fr)\}$ of sets of species $c_j(\fr)\subseteq X$. For
  $\mathbf{r}\in \RAF$, we shall further require that at least one set
  $c_{j^*}(\fr)$ consists only of species that are in the closure,
  i.e. ${c_{j^*}}(\fr)\subseteq \closure{\RAF}{\food}$. Any such choice of
  catalyzations renders $\RAF$ Reflexively Autocatalytic. It follows that
  such arbitrary assignments are infinite.
\end{proof}

\subsection{A preorder on RAFs}

The F-reactions of a RAF can be partially ordered such that for each
$x\in\feducts(\mathbf{r})$ we have $x\in\food$ or there is a reaction
$\mathbf{r}'\prec \mathbf{r}$ such that $x\in\fproducts(\mathbf{r}')$.  We
construct one such order that is reminiscent of a breadth-first search.  To
this end, we define sequences of nested species subsets $(X'_0, ..., X'_n)$
where $X'_0=\food$ and
\begin{equation}\label{eq:speciesSubSet}
  X'_i\coloneqq X'_{i-1}\cup \bigcup_{\feducts(\mathbf r)\subseteq X'_{i-1}}
  \fproducts(\mathbf r).
\end{equation}
This iterative process builds $\closure{\RAF}{\food}$ in a finite number of
steps as long as $\vert \RAF\vert
<\infty$.

These ordered sequences define a \emph{generation reaction-index}
$\gamma(\mathbf{r})$ where an F-reaction $\mathbf{r}$ happens for the first
time, i.e.
\begin{equation}\label{def:genreactionindex}
  \gamma(\mathbf{r})\coloneqq \min\{i \; | \; \feducts(\mathbf{r})\subseteq X'_i\} + 1,
\end{equation} 
where the index is shifted by one for normalization, so that the
\emph{first} generation corresponds to index $\gamma(\mathbf{r})=1$.
Similarly, the sequences define a \emph{generation species-index} $g(x)$
where a non-food species $x\not\in \food$ is produced for the first time,
i.e. \begin{equation} g(x)= \min \{i \; | \; x\in
X'_i\}.\end{equation} We further set $g(x)=0$ for $x\in\food$. Clearly, the 
total order $\le$ on the natural numbers induces a preorder $\preccurlyeq$ 
on $X\setminus \food$ and a preorder $\preceq$ on $\mathbf{R}$
\begin{equation}
\begin{cases}
  x_1\preccurlyeq x_2 \quad&\text{if}\quad \sgen(x_1) \leq \sgen(x_2);\\
  \fr_1\preceq \fr_2 \quad &\text{if} \quad \rgen(\fr_1) \leq \rgen(\fr_2).
\end{cases},
\end{equation}
which we refer to as \emph{generation orders} on species and F-reactions,
respectively. By Remark~\ref{rmk:RAFCRN}, we will use the index $\gamma$
indistinctly also for reactions \eqref{eq:Mod.SmEq} in a CRN setting. In
particular, from the definition of $\gamma$ it follows consistently that
$\gamma(r_1)=\gamma(r_2)$ whenever $r_1\sim r_2$ and thus
$\gamma(r_1)=\gamma(\fr)$ for $\fr=[r_1]_\sim$.

We conclude this section with the formal definition of CAFs sets.
\vspace{6pt}
\begin{definition}\label{def:CAF}
  A \emph{constructively autocatalytic and food-generated} set (CAF) is a
  RAF such that for each F-reaction $\fr$ all $x\in \feducts(\fr)$ satisfy
  $\sgen(x)<\rgen(\fr)$ and there is at least one catalyzation
  $c\in \cat(\mathbf{r})$ such that all F-catalysts $x\in c$ satisfy
  $\sgen(x)< \rgen(\fr)$.
\end{definition}

\subsection{Irreducible RAFs}\label{sec:IrrRAF}

For a CRS $(X,\FR,\cat(\FR))$ with food-set $\food\subseteq X$, 
the minimality of the RAF structure $\RAF$ with respect to its F-reactions 
has been investigated in \cite{steel_minimal_2013} via the 
notion of \emph{irreducible RAFs}.

\vspace{3mm}
\begin{definition}[Irreducible RAF]\label{def:irrerRAF}
  Let $\RAF$ be a RAF of the CRS $(X,\FR,\cat(\FR))$. $\RAF$ is
  \emph{irreducible} if no subset $\RAF'\subset \RAF$ determines a RAF.
\end{definition}

\emph{Example.} Consider a species set $X=\{x_0,...,x_{n+1}\}$ with food 
set $\food=\{x_0\}$ and a sequence of $n+1$ F-reactions
\begin{equation}\label{eq:ExReacSeq}
x_i + (x_{i+1})\quad\underset{\fr_i}\longrightarrow\quad x_{i+1}+(x_{i+1}),
\end{equation}
such that the F-product $x_{i+1}$ is also the F-catalyst (indicated in
brackets).  Clearly, any set $\RAF_k$ of the form
$\RAF_k\coloneqq \{\fr_0,\ldots, \fr_k\}$, with $k\leq n$ satisfies
Def.~\ref{def:RAF} and it is thus a RAF. In particular, $\RAF_0=\{\fr_0\}$
is an irreducible RAF.

The minimality condition imposed on the set of F-reactions via
Def.~\ref{def:irrerRAF} directly yields the following two structural lemmas.
\vspace{6pt}
\begin{lemma}\label{lem:OneFReacPerRAF}
  Let $(X,\FR,\cat(\FR))$ be a CRS with food-set $\food\subseteq X$. An
  irreducible RAF $\RAF$ never contains two distinct F-reactions $\fr_1$
  and $\fr_2$ such that either of the following conditions hold:
  \begin{enumerate}
  \item their F-reactants and F-products coincide: 
  \begin{equation}\label{eq:netreaction}
    \feducts(\fr_1)=\feducts(\fr_2) \qquad\text{and}\qquad
    \fproducts(\fr_1)=\fproducts(\fr_2) 
  \end{equation}
\item the F-reactants of $\fr_1$ coincides with the F-products of $\fr_2$,
  and vice versa:
  \begin{equation}\label{eq:inverse}
    \feducts(\fr_1)=\fproducts(\fr_2) \qquad\text{and}\qquad
    \feducts(\fr_2)=\fproducts(\fr_1) 
  \end{equation}
\end{enumerate}
\end{lemma}
\begin{proof}
  Indirectly assume there exist two F-reactions $\fr_1,\fr_2\in \RAF$ such
  that either \eqref{eq:netreaction} or \eqref{eq:inverse} holds and
  assume, without loss of generality, that $\fr_1\preceq \fr_2$. Consider
  the set $\RAF'\coloneqq \RAF \setminus \fr_2$. Since
  \eqref{eq:netreaction} (resp. \eqref{eq:inverse} holds, it follows that
  $\closure{\RAF}{\food} = \closure{\RAF'}{\food}$, since $\fr_1$ and
  $\fr_2$ have the same F-reactants and F-products (resp. since the
  F-products of $\fr_2$ are F-reactants of $\fr_1$). Therefore, $\RAF'$ is
  (i) Reflexively Autocatalytic, since $\RAF$ is, and (ii) is
  $\food$-generated; in particular, $\RAF'$ is a RAF which contradicts
  $\RAF$ being irreducible.
\end{proof}
We remark that \eqref{eq:inverse} in particular formally excludes the presence 
in any irreducible RAF of two F-reactions that are each other's reverse at the 
level of F-reactants and F-products. More specifically, for an F-reaction $\fr$, 
we write $\bar{\fr}$ for its reverse reaction satisfying 
$\feducts(\bar{\fr})=\fproducts(\fr)$ and $\fproducts(\bar{\fr})=\feducts(\fr)$. Similarly, for 
a set of F-reactions $
  \RAF' = \{\fr_1,\fr_2,\ldots,\fr_k\}$, we write $\bar{\FR}'' = \{\bar{\fr}_1,\bar{\fr}_2,\ldots,\bar{\fr}_k\}$ 
  for the set of the reverse F-reactions of $\RAF'$. We then state the following result.
\vspace{6pt}
\begin{lemma}\label{lem:noReverseIrrRAF}
Let $(X,\FR,\cat(\FR))$ be a CRS with food-set $\food\subseteq X$, and $\RAF \subset \FR$ 
an irreducible RAF with $|\RAF|>1$. For any subset $\RAF' \subseteq \RAF$, the set 
$\widetilde{\RAF}:=(\RAF \setminus \RAF') \cup \bar{\FR}''$ is not an irreducible RAF.
\end{lemma}
\begin{proof}
Indirectly assume that $\widetilde{\RAF}$ is an irreducible RAF and pick an F-reaction $\bar{\fr}_i$ as
\[\bar{\fr}_i\in
\operatorname{argmin}_{\bar{\fr}\in\bar{\FR}''}\gamma(\bar{\fr}),\]
where $\gamma$ is the generation order defined in \eqref{def:genreactionindex}.
By construction, $\rho_F(\bar{\fr}_i)\subseteq \closure{\RAF \setminus \RAF'}{\food}$, 
and thus $\pi_F(\fr_i)\subseteq \closure{\RAF \setminus \RAF'}{\food}$. Consequently, 
$\closure{\RAF}{\food}=\closure{\RAF \setminus \{\fr_i\}}{\food}$. Since $|\RAF|>1$, 
$\RAF\setminus\{\fr_i\}$ is nonempty and thus a RAF, contradicting the irreducibility of $\RAF$.
\end{proof}

In words, given an irreducible RAF $\RAF$, there is no subset of F-reactions $\RAF' \subseteq \RAF$ 
such that when all F-reactions in the subset $\RAF'$ are reversed, the resulting set is also an irreducible 
RAF. The only admissible exception is an irreducible RAF consisting of a single reaction where all 
F-reactants, F-products, and at least one catalyst are all in the food set $\food$.

Since a RAF is defined in Def.~\ref{def:RAF} as a subset of F-reactions,
Def.~\ref{def:irrerRAF} imposes a minimality condition only on the set of
F-reactions. However, motivated by the fact that in the CRN-setting
different catalyzations of the same F-reaction correspond to different
reactions, we also consider irreducible RAFs $\RAF$ with a minimal number
of catalyzations while conserving irreducibility.

For any $\food$-generated set $\RAF$, Lemma~\ref{lem:ArbCat} shows indeed
that infinitely many catalyzations exist that render $\RAF$ a
RAF. Complementarily, we address which catalyzation subsets
$\cat'(\FR)\subset \cat(\FR)$ preserve the fact that an $\food$-generated
set $\RAF$, which is an irreducible RAF in the CRS $(X,\FR, \cat(\FR))$, is
still an irreducible RAF in the CRS $(X,\FR, \cat'(\FR))$.  The following
Lemma shows that $\cat'(\FR)$ can always be chosen such that
$|\cat'(\fr)|=1$ for all F-reactions $\fr\in \RAF$.

\vspace{6pt}
\begin{lemma}\label{lem:mono}
  Let $\mathbf{\Phi}\coloneqq (X,\FR,\cat(\FR))$ be a CRS with specified
  \emph{food set} $\food\subseteq X$ and let $\RAF$ be an irreducible
  RAF. Then $\RAF$ is also an irreducible RAF in any CRS
  $(X, \FR,\cat'(\FR))$ where the catalyzations
  $\cat'(\FR)\subseteq\cat(\FR)$ satisfy:
  \begin{equation}\label{eq:casesmono}\cat'(\fr)=\begin{cases}
    \{c\}\text{ with $c\subseteq \closure{\RAF}{\food}$}\quad
    &\text{if $\fr\in \RAF$,}\\ 
    \cat(\fr)&\text{otherwise.}
  \end{cases}, 
\end{equation}
\end{lemma}
\begin{proof}
  The first case in \eqref{eq:casesmono} guarantees that $\RAF$ is
  Reflexively Autocatalytic in $(X, \FR,\cat'(\FR))$. Since
  $\cat'(\FR)\subseteq\cat(\FR)$, any subset of $\RAF$ that is a RAF in
  $(X,\FR,\cat'(\FR))$ is also a RAF in $(X,\FR,\cat(\FR))$, while the
  $\food$-generated property is independent of the catalyzations. Therefore, $\RAF$ is an irreducible RAF in
  $(X, \FR,\cat'(\FR))$.
\end{proof}

Lemma~\ref{lem:mono} suggests introducing the following notion.
\vspace{6pt}
\begin{definition}
  Let $(X,\FR,\cat(\FR))$ be a CRS with \emph{food set}
  $\food\subseteq X$. We call $\RAF$ a \emph{monocatalyzed} RAF if it is a
  RAF and it holds $|\cat(\fr)|=1$ for all F-reactions $\fr \in \RAF$.
\end{definition}

A monocatalyzed RAF may include, nevertheless, F-reactions
$\fr$, whose unique catalyzation $c$ is a set of multiple pairs
$(x_i, k^c_{x_i\fr})$, i.e.
$c=\{((x_1, k^c_{x_1\fr}),..., (x_n, k^c_{x_n\fr})\}$. Lemma \ref{lem:mono}
implies that, for a CRS $\mathbf{\Phi}=(X,\FR,\cat(\FR))$ with irreducible
RAF $\RAF\subseteq\FR$, there exists potentially exponentially many sub-CRS $\Phi'_i$
for which $\RAF$ is a monocatalyzed irreducible RAF. This is formalized in
the following lemma.

\vspace{6pt}
\begin{lemma}\label{lem:ExpoSubCRS}
  Let $\mathbf{\Phi}=(X,\FR,\cat(\FR))$ be a CRS with specified \emph{food
    set} $\food\subseteq X$ and $\RAF$ be an irreducible RAF. Then there
  exist at least $N$ distinct sub-CRS
  $\mathbf{\Phi}'_1,\ldots, \mathbf{\Phi}'_N$ with
  \begin{equation} \label{eq:ExpoNum} N := \prod_{\fr\in \RAF}\ \vert \{ c
    \in\cat(\fr) \; \mid\; c \subseteq \closure{\RAF}{\food} \vert,
  \end{equation}
  such that for $i=1,\ldots,N$ it holds that $\RAF$ is a monocatalyzed
  irreducible RAF in $\mathbf{\Phi}_i$.
\end{lemma} 
\begin{proof}
  Lemma \ref{lem:mono} implies that there are at least
  $N$ distinct combinations of
  possible monocatalyzations for all F-reactions in $\RAF$. Any such
  monocatalyzation $\cat'(\RAF)$ provides a sub-CRS
  $\mathbf{\Phi}_i'=(X,\FR,\cat'(\FR))$, where
  $\cat'(\RAF)\subseteq\cat(\FR)$ is arbitrarily extended to $\FR$, for
  which $\RAF$ is a monocatalyzed irreducible RAF.
\end{proof}

Finally, we underline that irreducibility is a property that is of interest
only for RAFs that are not CAFs: \vspace{6pt}
\begin{proposition}\label{prop:Caftrivial}
  Any irreducible RAF $\RAF$ that is itself a CAF is a single F-reaction
  $\RAF=\{\fr\}$.
\end{proposition}
\begin{proof}
  Let $\RAF$ be an irreducible RAF that is a CAF and consider any
  F-reaction $\fr\in \RAF$, which satisfies $\gamma(\fr)=1$. By
  Def.~\ref{def:CAF} and of generation reaction-index $\gamma$
  (Eq.~\eqref{def:genreactionindex}), it follows that all F-reactants
  $\rho(\fr)$ and at least one catalyzation $c\in \cat(\fr)$ are in the
  food set $\food$. Thus $\fr$ is a CAF (and RAF) itself, and
  $\RAF=\{\fr\}$ by irreducibility of $\RAF$.
\end{proof}

In turn, Example \eqref{eq:ExReacSeq} shows that we may have
single-reaction irreducible RAFs, which are not CAF. Due to
Prop.~\ref{prop:Caftrivial}, throughout the paper we exclude the trivial
CAF case whenever making statements about irreducible RAFs.  To keep the
exposition concise, we assume this exclusion throughout and do not repeat
it explicitly in each statement.

\subsection{The K\H{o}nig Graph of Monocatalyzed, Irreducible
  RAFs}\label{subsec:King}

Let $\mathbf{\Phi}=(X,\FR,\cat(\FR))$ be a CRS with food set
$\food\subseteq X$. We consider now the sub-CRS
$\mathbf{\Phi}'\subseteq \mathbf{\Phi}$ induced by a monocatalyzed
irreducible RAF $\RAF$, i.e.
\begin{equation}\label{eq:RAF-CRS}
  \mathbf{\Phi}' \coloneqq (\closure{\RAF}{\food}, \RAF, \cat'(\RAF)),
\end{equation}
where $|\cat'(\fr)|=1$ for each F-reaction $\fr\in \RAF$. In the following,
without loss of generality, we assume that each food species $x\in\food$
participates as an F-reactant, F-product, or F-catalyst in at least one
F-reaction $\fr\in\RAF$. This excludes isolated food species for
simplicity, as they play no role in the subsequent constructions and
statements.

Consider now the associated CRN $\Gamma'\coloneqq
\Gamma(\mathbf{\Phi}')$. One central advantage of considering such
monocatalyzed CRSs $\mathbf{\Phi}'$ is that the quotient map from
$\Gamma(\mathbf{\Phi}')$ to $\mathbf{\Phi}'$ is a bijection. Indeed, the
quotient construction becomes trivial, since every equivalence class,
i.e. every F-reaction $\fr$, is a singleton $\fr=\{r\}$. Thus, the quotient
CRS can be canonically identified with the underlying CRN and we can
therefore naturally address properties of $\mathbf{\Phi}'$ by studying them
in $\Gamma'$. To do so, we use the notation $\mathcal{R}'$ to identify the
CRN reactions of $\Gamma'$,
i.e. $\Gamma'=(\closure{\RAF}{\food},\mathcal{R}')$. We repeat that each of
such reactions $r\in \mathcal{R}'$, written as \eqref{eq:Mod.SmEq}, is in
canonical bijection with an F-reaction $\fr\in \RAF$, written as
\eqref{eq:FReacEq}, which is catalyzed by the singleton $\cat(\fr)=\{c\}$
with $c=\cat(r)$.

We may consider now the K{\H{o}}nig graph $K(\RAF)=K(\Gamma')$ of the
monocatalyzed RAF $\RAF$ simply as the one of its associated CRN
$\Gamma'=(\closure{\RAF}{\food},\mathcal{R}')$. This is the bipartite graph
$\king(\Gamma') \coloneqq (\closure{\RAF}{\food} \cup \CR', E)$, whose edge
set $E\coloneqq E_1\cup E_2$, is partitioned into two distinct types of
directed edges given as
\begin{equation}
  \begin{cases} (x,r) \ | \ x\in \closure{\RAF}{\food},
    r\in \CR': s^-_{xr}>0 \\
    (r, x) \ | \ x\in \closure{\RAF}{\food},
    r\in \CR': s^+_{xr}>0.
  \end{cases}
  \label{eq:Koenig}
\end{equation}
The K\H{o}nig graph emerged under various names in CRN theory, as the
\emph{species-reaction} (SR) graph \cite{craciun_multiple_2006} and
\emph{Petri nets} \cite{angeli_petri_2007}. We note that a K{\H{o}}nig
  graph can be constructed from any reaction network
  $\Gamma=(X,\mathcal{R})$ in the same manner, using $x\in X$ and
  $r\in\mathcal{R}$ in Eq.~\eqref{eq:Koenig}.

A \emph{directed path} $\mathcal{P}$ in $K(\RAF)$ is an ordered sequence of
vertices $\mathcal{P}=(v_1,\ldots, v_n)$, starting at $v_1$ and terminating
at $v_n$, such that $(v_i,v_{i+1})\in E(\king(\RAF)), 1\leq i\leq n-1$, and
no vertex is repeated. We say that such $\mathcal{P}$ is a path from $v_1$
to $v_n$.  $\king(\RAF)$ is \emph{strongly connected} if, for every pair of
vertices $v,w\in (X\cup\mathcal{R})$ there is a directed path from $v$ to
$w$. If $\king(\RAF)$ is not strongly connected, the \emph{strongly
  connected component} containing a vertex $v_1$ is the maximum set of
vertices $V_1$ with $v_1\in V_1$ such that for any pair of vertices $v$ and
$w$ in $V_1$ there is a directed path from $v$ to $w$.

We note that, for a species $x_f \in \food$ in the food set, there may be
no directed edge $(r,x_f)$ at all.  Complementarily, we introduce the set
$\mathbf{W}_{\RAF}$ of \emph{waste species} (with respect to $\RAF$) as the
species $x_w$ for which there is no directed edges $(x_w,r)\in
\king(\RAF)$. Equivalently, in symbols:
\begin{equation}\label{eq:Waste}
  \mathbf{W}_{\RAF}\coloneqq \{x\in \closure{\RAF}{\food} \ | \
  \nexists r \in \CR' \text{ s.t. } s^-_{xr}>0\}.
\end{equation}
We now have the following statement:
\vspace{6pt}
\begin{lemma}[Strong connectivity]\label{lem:PSC}
  Let $\RAF$ be an irreducible, monocatalyzed RAF, and let $v_1,v_2$ be two
  vertices of $\king(\RAF)$, which are not in the food and waste sets, i.e.
  \begin{equation}
    v_1,v_2\in (\closure{\RAF}{\food} \cup \mathcal{R}') \setminus
    (\food \cup\mathbf{W}_{\RAF}).
  \end{equation} 
  Then there is a directed path $\mathcal{P}$ from $v_1$ to $v_2$ such that
  $\mathcal{P}\cap \food =\emptyset$.
\end{lemma}
\begin{proof}
  Indirectly assume that there is no path $\mathcal{P}$ from $v_1$ to
  $v_2$. If $v_1$ is a species-vertex, $v_1=x_1$, then
  $v_1\not\in\mathbf{W}_{\RAF}$, and let $r_1$ be any reaction vertex such
  that $s^-_{x_1r_1}>0$. By indirect assumption, there is no path from
  $r_1$ to $v_2$. Thus, we may always reduce to the case of $v_1=r_1$ being
  a reaction vertex. The reachable set of \emph{descendant vertices} from a
  vertex $u$ is denoted as
  \begin{equation}
    \mathbf{D}(v_1) \coloneqq \{v \in V \ \vert \
    \text{there exists path } \mathcal{P}=(u, \ldots, v)\}.
  \end{equation}
  Consider now $\RAF_1$ as the set of F-reactions in canonical bijection
  with the reactions
  $\mathcal{R}'_1\coloneqq \mathcal{R}' \setminus \mathbf{D}(r_1)$. Our
  indirect assumption implies that $v_2\not\in\mathbf{D}(r_1)$, and thus it
  follows that, for any F-reaction $\fr\in \RAF_1$,
  $\cat'(\fr)\cup \rho_F(\fr)\subseteq \closure{\RAF_1}{\food}$. In
  particular, $\RAF_1$ satisfies the RAF Def.~\ref{def:RAF}, which
  contradicts the irreducibility of $\RAF$: $\RAF_1\subset \RAF$ is a
  strict subset of $\RAF$ since it does not contain
  $\mathbf{r}_1$. Finally,
  $\kappa'_F(\fr)\cup \rho_F(\fr)\subseteq \closure{\RAF_1}{\food}$ holds
  for any $\fr\in \RAF_1$ also whenever all paths $\mathcal{P}$ from $v_1$
  to $v_2$ are such that $\mathcal{P}\cap \food \neq \emptyset$, thus the
  full statement follows again by contradiction.
\end{proof}

Lemma \ref{lem:PSC} is essentially a strong-connectivity result on
$K(\RAF)$ without food and waste species.  However, strongly connected
graphs may still have \emph{cut-vertices}, namely, vertices whose removal
disconnects the underlying undirected graph (see the glossary on graph
theory SI Sec.~\ref{sec:Glossary}). The example in
Fig.~\ref{ex:CutVertex2} captures this situation, as each reaction of the
associated CRN of an irreducible, monocatalyzed RAF is a cut-vertex in the
bipartite K\H{o}nig graph $\king(\RAF)$. The next result shows that,
indeed, for $\king(\RAF)$ cut-vertices can only be reaction-vertices
$r\in \mathcal{R}'$.

\vspace{6pt}
\begin{lemma}[Cut vertices]\label{lem:cutvert}
  Let $\RAF$ be an irreducible, monocatalyzed RAF. Every cut-vertex of
  $\king(\RAF)$ is a reaction $r\in \mathcal{R}'$.
\end{lemma}
\begin{proof}
  Indirectly, assume that there exists a cut vertex $x^*$ of $\king(\RAF)$
  that is a species. Consider now the connected components
  $\tilde{K}_1,...,\tilde{K}_n$ of the bipartite subgraph $\tilde{K}$
  obtained from $\king(\RAF)$ by removing $x^*$ and all adjacent edges
  $(x^*,r)$ or $(r,x^*)$. Note that - by construction - no species-vertex
  in any component $\tilde{\king}_i$ is an F-reactant or an F-catalyst for any
  reaction-vertex in any component $\tilde{\king}_j$, with $i\neq j$. In
  particular, from $\RAF$ being a RAF, it follows that there is at least
  one such connected component $\tilde{\king}_1$ whose associated F-reaction
  set $\RAF_1$ is a RAF as well. Since any other
  $\tilde{\king}_i \neq \tilde{\king}_1$ contains, by construction, a
  reaction-vertex, then the RAF $\RAF_1$ is a strict subset of $\RAF$,
  which leads to a contradiction.
\end{proof}

Building on the perspectives of Lemmata \ref{lem:PSC} and
\ref{lem:cutvert}, the third and main result of this section considers
\emph{maximal strong blocks} $\block$ of $\king(\RAF)$, i.e. maximal
strongly-connected subgraphs that possess no cut-vertices
\cite{grotschel_minimal_1979}.  An example of an irreducible, monocatalyzed
RAF containing multiple strong blocks is captured in
Fig.~\ref{ex:CutVertex2}.

\vspace{6pt}
\begin{lemma}[Strong blocks]\label{lem:StrongBlocks}
  Consider an irreducible, monocatalyzed RAF $\RAF$. Let $\block$ be a
  maximal strong block of $\king(\RAF)$ such that $\block$ contains a
  species-vertex $x^*$, which is not in the food and waste sets,
  i.e. $x^*\notin \food \cup\mathbf{W}_{\RAF}$, and let
  $r^*\in \mathcal{R}(\block)$ be a reaction satisfying $r^* \preceq r$,
  for any $r\in \mathcal{R}(\block)$. Then,
  \begin{itemize}
  \item[i.] All F-reactants of $r^*$ in $\block$ belong to the food
    set:{$$\feducts(r^*)\cap X(\block) \quad \subseteq \quad \food;$$}
  \item[ii.] There is at least a non-food F-catalyst of $r^*$ in $\block$:
    $$(\cat(r^*) \setminus \food) \cap X(\block)\quad \neq \quad \emptyset;$$
  \item[iii.] There is at least a non-food F-product of $r^*$ in
    $\block$:$$(\fproducts(r^*) \setminus \food) \cap X(\block)\quad \neq
    \quad \emptyset.$$
  \end{itemize} 
\end{lemma}
\begin{proof}
  (i) Indirectly assume there is a species vertex $x\not\in \food$ in
  $\block$ such that $x$ is an F-reactant of $r^*$. Lemma \ref{lem:cutvert}
  excludes that $x$ is a cut-vertex for $\king(\RAF)$, and therefore all
  reactions $\tilde{r}\in \mathcal{R}'$ for which $x$ is an F-product must
  be vertices in $\block$, contradicting the assumed $\preceq$-minimality
  of $r^*$ in $\block$.
  \\
  (ii) Since $x^*$ is a species vertex of the maximal strong block
  $\block$, it follows that any directed path from $r^*$ to $x^*$ and from
  $x^*$ to $r^*$ touches only vertices in $\block$ (ear-decomposition of
  strong blocks). Via Lemma \ref{lem:PSC}, there exists at least one
  directed path $\mathcal{P}$ from $x^*$ to $r^*$ such that
  $\mathcal{P}\cap \food =\emptyset$, i.e. there exists species-vertex
  $x\not\in\food$ in $\block$ which is either an F-reactant or an
  F-catalyst of $r^*$. Statement (i), proved above, excludes that such $x$
  is an F-reactant, so (ii) follows.
  \\
  (iii) preliminarily observe that $r^*$ possesses an F-product
  $x_\pi \not\in \food$, since otherwise $\RAF \setminus \{r^*\}$ would still
  be a RAF, contradicting irreducibility of $\RAF$. Let then $x_\pi$
  indicate an F-product of $r^*$ and $x_\kappa$ and F-catalyst of $r^*$ in
  $\block$, whose existence we proved above in (ii). Assume now indirectly
  that (iii) does not hold, then $x_\pi\neq x_\kappa$ and there is no
  directed path $\mathcal{P}=(r^*, x_\pi,...,x_\kappa)$ since otherwise
  $x_\pi$ would be in $\block$ as well, against the indirect
  assumption. Consider any reaction $\tilde{r}$ which has $x_\kappa$ as
  F-product and construct a family of sets iteratively
  $\tilde{\mathbf{R}}_n$ of F-reactions from
  $\tilde{\mathbf{R}}_0=\{\tilde{\fr}\}$ as follows:
  \begin{equation}\label{RAFbuild}
    \tilde{\mathbf{R}}_{n}=\tilde{\mathbf{R}}_{n-1}\cup
    \left\{\fr\in \RAF \; \middle| \; \bigcup_{\hat{\fr}\in \tilde{\mathbf{R}}_{n-1}}
    \bigg(\feducts(\hat{\fr}) \cup \cat(\hat{\fr})\bigg)
    \bigcap \left( \fproducts(\fr) \setminus \food \right) \neq \emptyset \right\}.
  \end{equation}
  Since $\RAF$ is a finite set, the iteration converges at a finite step
  $\bar{n}$ to a set $\RAF'\coloneqq \tilde{\mathbf{R}}_{\bar{n}}$.  By
  construction and the fact that $\RAF$ is a RAF, $\RAF'$ is a RAF
  itself. However, no direct path
  $\mathcal{P}=(r^*, x_\pi,..., \tilde{r},x_\kappa)$ exists, and thus
  $\fr^*\not\in \RAF'$, in contradiction with the assumption of
  irreducibility of $\RAF$.
 \end{proof}

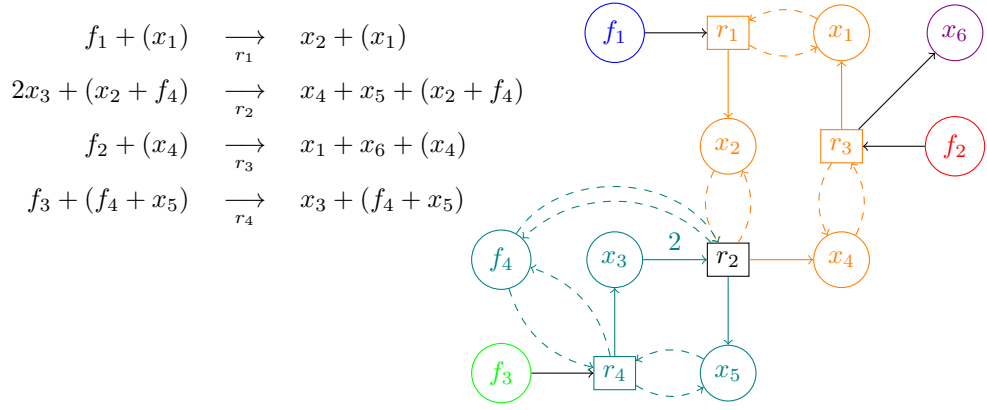
\begin{figure}[bth]
  \begin{minipage}[c]{0.45\textwidth}
  \vspace{-72pt}
    \begin{equation*}
      \begin{split}
	f_1 + (x_1) \quad&\underset{r_1}\longrightarrow \quad x_2 + (x_1) \\
	2x_3 + (x_2 + f_4) \quad&\underset{r_2}\longrightarrow\quad x_4 + x_5 + (x_2 + f_4) \\	
	f_2 + (x_4)\quad&\underset{r_3}\longrightarrow \quad x_1 + x_6 + (x_4)\\
	f_3 + (f_4 + x_5)\quad &\underset{r_4}\longrightarrow \quad x_3 + (f_4 + x_5)\\	
      \end{split}
    \end{equation*}
  \end{minipage}%
  \hfill%
  \begin{minipage}[c]{0.55\textwidth}
  \centering
    \begin{tikzpicture}
      \node[draw, circle, blue] (f1) at (-1.5,0) {$f_1$};      
      \node[draw, circle, red] (f2) at (3,-1.5) {$f_2$};
      \node[draw, circle, green] (f3) at (-3,-4.5) {$f_3$};
      \node[draw, teal, circle] (f4) at (-3,-3) {$f_4$};
            
      \node[draw, orange, circle] (x1) at (1.5,0) {$x_1$};			
      \node[draw, orange, circle] (x2) at (0,-1.5) {$x_2$};      
      \node[draw, teal, circle] (x3) at (-1.5,-3) {$x_3$};
      \node[draw, orange, circle] (x4) at (1.5,-3) {$x_4$};
      \node[draw, teal, circle] (x5) at (0,-4.5) {$x_5$};
      \node[draw, circle, violet] (x6) at (3,0) {$x_6$};
      
      \node[draw, orange, rectangle] (r1) at (0,0) {$r_1$};
      \node[draw, rectangle] (r2) at (0,-3) {$r_2$};
      \node[draw, orange, rectangle] (r3) at (1.5,-1.5) {$r_3$};
       \node[draw, teal, rectangle] (r4) at (-1.5,-4.5) {$r_4$};
      
      \draw[->] (f1) to (r1);
      \draw[->, orange, dashed] (x1) to [bend right = 30] (r1);
      \draw[->, orange, dashed] (r1) to [bend right = 30] (x1);
      \draw[->, orange]  (r1) to (x2);
      \draw[->, orange, dashed] (x2) to [bend right = 30] (r2);	
      \draw[->, orange, dashed] (r2) to [bend right = 30] (x2);
      \draw[->, teal] (x3) to node[midway, above] {$2$} (r2);
      \draw[->, orange] (r2) to (x4);
      \draw[->, teal ] (r2) to (x5);
      \draw[->, orange, dashed] (x4) to [bend right = 30] (r3);
      \draw[->, orange, dashed] (r3) to [bend right = 30] (x4);
      \draw[->] (f2) to (r3);
      \draw[->, orange] (r3) to (x1);
      \draw[->] (r3) to (x6);
      \draw[->, teal, dashed] (x5) to [bend right = 30] (r4);
      \draw[->, teal, dashed] (r4) to [bend right = 30] (x5);      
      \draw[->] (f3) to (r4);
      \draw[->, teal, dashed] (f4) to [bend right = 30] (r4);
      \draw[->, teal, dashed] (r4) to [bend right = 30] (f4);
      
      \draw[->, teal, dashed] (f4) to [bend left = 60] (r2);
      \draw[->, teal, dashed] (r2) to [bend right = 45] (f4);      
      \draw[->, teal] (r4) to (x3);      
    \end{tikzpicture}
  \end{minipage}
  \caption{\textit{(Left):} Example of an irreducible, monocatalyzed RAF
    written as CRN reactions with a cut-vertex that has a single
    catalyst. Catalyzations are indicated in brackets. \textit{(Right):}
    Depiction of the corresponding König graph with catalyzations indicated
    as dashed lines. Different strong blocks are depicted in different
    colours, while the cut-vertices $r_2$, belonging to two strong blocks
    (teal and orange), is coloured black. For the teal strong block, $r_4$
    is the `$r^*$' reaction expected via Lemma \ref{lem:StrongBlocks} and
    for the orange strong block it is $r_1, r_2$, and $r_3$.}
  \label{ex:CutVertex2}
\end{figure}

We refer to a reaction $r^*$ as in the statement of
Lemma~\ref{lem:StrongBlocks} as an \emph{entry reaction} for $\block$. By
an inspection of the example in Fig.~\ref{ex:CutVertex2}, we observe
that such entry reactions appear in intimate relation with cut-vertices of
the bipartite graph $K(\mathcal{R}')$. We investigate this intuition in the 
following lemma, which generalizes part (iii) of Lemma~\ref{lem:StrongBlocks} 
to reactions $\tilde{r}$ that are cut vertices, but not necessarily minimal in the 
generation order.
\vspace{6pt}

\begin{lemma}\label{cor:strongblock1}
 Let $\RAF$ be an irreducible monocatalyzed RAF. Let $\block_1, \block_2$ be two maximal
 strong blocks of $K(\RAF)$ containing at least one species-vertex outside
 the food and waste set, which intersect in a reaction $\tilde{r},$ i.e. $\block_1\cap\block_2=\{\tilde{r}\}$. 
 Then $\tilde{r}$ has at least one F-product in $\block_1$ and one F-product in $\block_2$, i.e.
 \[\fproducts(\tilde{r})\cap X(\block_i)\neq\emptyset,\qquad i=1,2.\]
 In particular, $\tilde{r}$ has at least two distinct non-food F-products.
\end{lemma}

\begin{proof}
For $i\in{1,2}$, pick a species $\tilde{x}_i\in X(\block_i)\setminus(\food\cup\mathbf{W}_{\RAF})$. 
Via Lemma~\ref{lem:PSC}, there exist directed paths $\mathcal{P}_i$ in $K(\RAF)$ from $\bar r$ 
$\tilde{r}$ to $\tilde{x}_i$ such that $\mathcal{P}_i\cap\food=\emptyset$. 
Since $\block_i$ are maximal strong blocks, it follows that each of the paths is contained in one of the strong blocks, 
specifically:  $\mathcal{P}_i\subseteq\block_i$. Let now $\tilde{\tilde{x}}_{i}$ be the first 
species-vertex after $\tilde r$ in $\mathcal{P}_i$. If $\{\tilde{r}, \tilde{\tilde{x}}_i\}$ 
are the only two vertices in $\block_i$, then $\tilde{\tilde{x}}_i$ is an F-product 
of $\tilde{r}$. Otherwise, there exists at least another reaction $\tilde{\tilde{r}}_i\in \mathcal{R}(\block_i)$, 
$\tilde{\tilde{r}}_i\neq \tilde{r}$, which has $\tilde{\tilde{x}}_i$ as an F-product. Indirectly assume there exists no 
non-food F-product $x_{\pi_i}$ of $\tilde{r}$ in $\block_i$. Arguing exactly as in the proof of Lemma \ref{lem:StrongBlocks}, starting from
$\tilde{\mathbf{R}}_0=\{\tilde{\tilde{\fr}}\}$ and building a RAF via \eqref{RAFbuild}, we construct a nonempty 
$\RAF'\subseteq \RAF\setminus \{\tilde{\fr}\}\subset \RAF$, leading to contradiction.   
\end{proof}

Lemma~\ref{cor:strongblock1} directly yields the following statement,
which gives an easy-to-check structural necessary condition for the
presence of multiple strong blocks involving species that are not in the
food or in the waste sets.  \vspace{6pt}
\begin{corollary}\label{cor:strongblock2}
  Let $\RAF$ be an irreducible monocatalyzed RAF. If every F-reaction
  $\fr\in\RAF$ satisfies
  \begin{equation}\label{eq:condmanysb}
    |\fproducts(\fr)|<2,
  \end{equation}
  then $K(\RAF)$ has a unique maximal strong block containing at least one
  species-vertex outside the food and waste sets.
\end{corollary}
\begin{proof}
Indirectly assume that there exist two distinct maximal strong blocks 
$\block_1$ and $\block_2$, each containing a species-vertex outside the food 
and waste sets. Via Lemma~\ref{lem:PSC}, there exists a directed path between 
such species in $\block_1$ and $\block_2$ which does not intersect the food set. 
Such a path must cross at least one cut-vertex via Lemma \ref{lem:cutvert} 
of two maximal strong blocks each containing a species-vertex 
outside the food and waste set, leading via Lemma \ref{cor:strongblock1} to a 
contradiction with the assumption.
\end{proof}
For example, if all F-reactions $\fr\in\RAF$ are of the form
\begin{equation*} 
  x_\rho + (x_\kappa)\quad\underset{\fr}{\longrightarrow}\quad
  x_\pi + (x_\kappa)
\end{equation*}
then \eqref{eq:condmanysb} is satisfied and thus $K(\RAF)$ is essentially
one single strong block. Here, `essentially' means by not
considering strong blocks involving only food and waste species.

Finally, we stress that entry reactions always identify a column in the
associated stoichiometric matrix
$S[X(\block)\setminus \food, \mathcal{R}(\block)]$ of
$\Gamma'=(\closure{\RAF}{\food},\mathcal{R}')$, restricted to the non-food
species in $\block$, which is nonnegative (property (i) in
Lemma~\ref{lem:StrongBlocks}) and it has at least one positive entry
(property (ii) in Lemma~\ref{lem:StrongBlocks}). Such an observation will
be pivotal to prove stoichiometric autocatalysis in each of such strong
blocks $\block$, as per our main result below. We proceed first by
introducing formally stoichiometric autocatalysis.

\section{Stoichiometric Autocatalysis}\label{sec:AutoCatIrrRAF}

In chemical reaction networks (CRN), a matrix-based definition of
\emph{stoichiometric autocatalysis} was introduced by Blokhuis et
al. \cite{blokhuis_Universal_2020} for networks without catalysis, and was
recently extended to the catalytic setting considered in the present paper
\cite{golnik_autocatalytic_2026}, where a reactant may also be a product of
the same reaction. The definition is as follows.

\vspace{6pt}
\begin{definition}\label{def:autSubRN}
  A CRN $\Gamma\coloneqq (X,\CR)$ is \emph{stoichiometric autocatalytic} if
  it contains a sub-CRN $\Gamma'=(X',\CR')$ such that
  \begin{enumerate}
  \item[i.] \emph{Well-formedness:} For any $r\in \mathcal{R}'$, there
    exists $x_1, x_2 \in X'$, not necessarily distinct, such that
    $s^-_{x_1r}s^+_{x_2r}\neq 0$;
  \item[ii.] \emph{Semi-positivity:} For the associated stoichiometric
    matrix $\SM'\coloneqq \SM[X',\CR']$ there exists
    $v\in \mathbb{R}_{>0}^{|\CR'|}$ such that $\SM'v>0$.
  \end{enumerate}
\end{definition}
\vspace{6pt}

A minimality concept for stoichiometrically autocatalytic CRN is expressed
in terms of subnetworks. A sub-CRN $\Gamma'$ that satisfies \textit{(i.)}
and \textit{(ii.)}, while no proper sub-CRN
$\Gamma''\subset\Gamma'$ 
does, is called an \emph{autocatalytic core}.  Any autocatalytic sub-CRN
thereby contains an autocatalytic core. The stoichiometric submatrix
associated to an autocatalytic core admits a reordering of the columns such
that all diagonal entries are nonpositive while all off-diagonal entries
are nonnegative \cite{blokhuis_Universal_2020, golnik_autocatalytic_2026, 
vassena_unstable_2024}. Matrices with
nonnegative off-diagonal entries are called \emph{Metzler matrices} in the
literature. An autocatalytic core of size $1$, i.e.,
$|X'|=|\mathcal{R}'|=1$, is thus simply an \emph{autocatalytic reaction}
$r$ of the type:
\begin{equation}
    s^-_{x_1r}x_1+\ldots \quad\underset{r}\longrightarrow \quad s^+_{x_1r}x_2+\ldots\quad\quad\quad\text{with $s^+_{x_1r}>s^-_{x_1r}$.}
\end{equation}
By this simple example, it appears evident the quantitative role of the
stoichiometric coefficients in the definition of stoichiometric
autocatalysis (via the semipositivity condition). In RAF, in turn, the
quantitative value of the stoichiometry is irrelevant. However,
semipositivity of the stoichiometric matrix restricted to the non-food
molecules follows naturally for any food-generated set of reactions, as the
next lemma shows. We recall Remark \ref{rmk:RAFCRN} on the validity 
of the RAF definition for CRN.

\vspace{6pt}
\begin{lemma}~\label{lem:NonFoodSemipositiv2}
  Let $\mathcal{R}'$ be a food-generated set of reactions. Let $
  n:=\max\{\gamma(r)\mid r\in\CR'\}$, and for $k=1,\ldots,n$, define
  \begin{equation*}
    \mathcal{R}^k
    \coloneqq
    \{r\in\mathcal{R}'\mid \gamma(r)=k\},
    \qquad
    X^k
    \coloneqq
    \{x\in \closure{\mathcal{R'}}{\food}\setminus\food\mid \sgen(x)=k\}.
  \end{equation*}
  Pick any
  $\lambda=(\lambda_n,\ldots,\lambda_1)$ iteratively by choosing
  $\lambda_n>0$ and, for $k=n-1,\ldots,1$,
  \begin{equation*}
    \lambda_k
    > \max \left\{
      \max_{x\in X^k}
      \left[
        -\frac{
          \displaystyle
          \sum_{\ell=k+1}^{n}
          \lambda_\ell
          \sum_{r\in\mathcal{R}^\ell}\SM_{xr}
        }{
          \displaystyle
          \sum_{r\in\mathcal{R}^k}\SM_{xr}
        }
      \right],0 \right\},
  \end{equation*}
  and define $v\in \mathbb{R}^{|\mathcal{R}'|}_{>0}$ by $
  v_r\coloneqq \lambda_{\gamma(r)}.$
  Then
  \begin{equation*}
    \SM[\closure{\mathcal{R'}}{\food}\setminus\food,\mathcal{R}']v>0.
  \end{equation*}
  In particular,
  $\SM[\closure{\mathcal{R'}}{\food}\setminus\food,\mathcal{R}']$ is
  semi-positive.
\end{lemma}

\begin{proof}
  Arbitrarily take \(x\in \closure{\mathcal{R}'}{\food} \setminus\food\),
  and let \(k=\sgen(x)\). By the definition of species generation,
  $ \SM_{xr}=0$ for every $r\in\mathcal R^\ell,$ with $\ell<k,$ while
  $ \sum_{r\in\mathcal R^k}\SM_{xr}>0.$ Since \(v_r=\lambda_{\gamma(r)}\),
  it follows that
  \begin{equation*}
    \bigl(\SM[X'\setminus\food,\mathcal R']v\bigr)_x
    =
    \sum_{r\in\mathcal R'}\SM_{xr}v_r=
    \sum_{\ell=k}^{n}
    \lambda_\ell
    \sum_{r\in\mathcal R^\ell}\SM_{xr}
    =
    \lambda_k
    \sum_{r\in\mathcal R^k}\SM_{xr}
    +
    \sum_{\ell=k+1}^{n}
    \lambda_\ell
    \sum_{r\in\mathcal R^\ell}\SM_{xr}.
  \end{equation*}
  The defining condition on \(\lambda_k\) in the assumption of the Lemma
  directly yields
  \begin{equation*}
    \bigl(\SM[X'\setminus\food,\mathcal R']v\bigr)_x>0.
  \end{equation*}
  Since \(x\in X'\setminus\food\) is arbitrary,
  $\SM[X'\setminus\food,\mathcal R']v>0$. Since $v>0$ by construction,
  \(\SM[X'\setminus\food,\mathcal R']\) is semi-positive.
\end{proof}

Lemma~\ref{lem:NonFoodSemipositiv2} shows that semipositivity is a natural
property also for RAFs. Moreover, the structural result
Lemma~\ref{lem:StrongBlocks} yields the following theorem: each sub-CRN
$(X(\block)\setminus \food, \CR(\block))$, corresponding to a maximal,
non-trivial strong block $\block$ of $\king(\RAF)$ of an irreducible,
monocatalyzed RAF $\RAF$, is \emph{stoichiometrically autocatalytic}.

\vspace{6pt}
\begin{theorem}\label{Thm:StrBlockCore}
  Let $\RAF$ be an irreducible, monocatalyzed RAF, and let $\block$ be a
  maximal strong block in $\king(\RAF)$ containing at least one
  species-vertex $x^*$ that is not in the food and waste sets. Then
  $\block$ contains an autocatalytic core.
\end{theorem}
\vspace{6pt}

\begin{proof}
  We show that the subnetwork identified by $\block$ is stoichiometrically
  autocatalytic.
  Well-formedness follows naturally from the fact that any reaction-vertex
  $r$ in a strong-block $\block$ has at least one ingoing edge $(x_1,r)$
  and one outgoing edge $(r,x_2)$, with $x_1$ not necessarily distinct from
  $x_2$ but both $x_1,x_2\not\in \food \cup\mathbf{W}_{\RAF}$ via Lemma
  \ref{lem:StrongBlocks}.
  To show semi-positivity, consider an enlarged food
  set \begin{equation}\tilde{\food}:=\food\cup\{x\in
    \closure{\RAF}{\food}\;|\; x\not\in X(\block)\}.
  \end{equation}
  Lemma \ref{lem:StrongBlocks} guarantees that the set
  $\mathcal{R}(\block)$ is $\tilde{\food}$-generated as any entry reaction
  (with respect to the food set $\tilde{\food}$) $r^*$ has F-reactants only
  in $\tilde{\food}$. Since
  $\closure{\mathbf{R}'(\block}{\tilde{\food}})=X(\block)\setminus
  \tilde{\food}$, Lemma~\ref{lem:NonFoodSemipositiv2} yields semipositivity
  of $S[X(\block)\setminus \tilde{\food},\mathcal{R}(\block)]$ .
\end{proof}

\section{Examples}
\label{sec:examples}

\subsection{Irreducible RAFs from the Binary Polymer Model}

We started this contribution with the observation of a large number of
autocatalytic cores in irreducible RAFs from two instances of the Binary
Polymer Model. The first instance was composed only of cleavage and
ligation, while the second contained only ligation F-reactions; see
Fig.~\ref{fig:irrRAFCores}. To clarify this relation, we investigated the
structure of irreducible RAFs. We comment here on three aspects.
\begin{itemize}
\item \textbf{F-products in the food set.} The RAF Def.~\ref{def:RAF} and
  the consequent structural analysis in Sec.~\ref{subsec:King} do not
  distinguish F-products that belong to the food set, as our structural
  results disregard paths passing through food species. In contrast,
  Def.~\ref{def:autSubRN} of stoichiometric autocatalysis is sensitive to
  the presence of food species in an autocatalytic sub-CRN. Filtering out
  food species revealed that food-containing cores indeed contribute to the
  total number of autocatalytic cores, at least in irreducible RAFs
  containing both ligation and cleavage reactions
  (Fig.~\ref{fig:FoodVsNonFood}). That irreducible RAFs composed only of
  ligation reactions exhibit the same number of autocatalytic cores after
  removal of food species is, in turn, not surprising, since irreducibility
  excludes ligation F-reactions whose F-product lies in the food
  set. Nevertheless, many autocatalytic cores remain after food species are
  removed.
\item \textbf{The number of distinct catalyzations.}
  Lemma~\ref{lem:ExpoSubCRS} shows that a single `polycatalyzed'
  irreducible RAF may be associated with an exponentially large number of
  monocatalyzed RAFs. Since each such realization is itself an irreducible
  RAF, our results guarantee the presence of autocatalytic cores in each of
  the corresponding CRNs. This, however, does not imply per se an
  exponentially large number of \emph{distinct} autocatalytic cores: the
  correspondence clearly does not need to be injective at the level
  of autocatalytic cores. Multiple catalyzations should therefore be
  regarded only as a possible source of the observed
  multiplicity. Moreover, F-reactions with multiple catalyzations are not
  prevalent in the irreducible RAFs considered here, where only a handful
  of F-reactions admit more than one catalyzation.
\item \textbf{Strong blocks.} The results in Sec.~\ref{subsec:King} show
  that, for a monocatalyzed irreducible RAF, the subgraph of the associated
  K{\H{o}}nig graph obtained after excluding food and waste species is
  strongly connected. Moreover, the occurrence of multiple maximal strong
  blocks is much constrained: only reaction vertices may be cut-vertices,
  and their structure is restricted by Lemma~\ref{lem:StrongBlocks} and
  Lemma~\ref{cor:strongblock1}. For ligation-only instances of the Binary
  Polymer Model, for example, the presence of a single relevant maximal
  strong block follows directly from Cor.~\ref{cor:strongblock2}, since
  every F-reaction has a single F-product. Moreover, although multiple
  blocks are in theory possible (see SI
  Fig.~\ref{fig:BPMMultipleStrongBlocks}), even for BPM instances with
  maximum polymer lengths $4$, $5$, $6$, $7$, and $8$ containing both
  cleavage and ligation reactions, none of the collected irreducible RAFs
  exhibited more than one such block (SI
  Fig.~\ref{fig:StrongBlocks}). Thm.~\ref{Thm:StrBlockCore} shows that
  every relevant maximal strong block is stoichiometrically autocatalytic
  and therefore contains an autocatalytic core. The decomposition into
  multiple strong blocks could, however, not explain the large number of
  cores observed here, since the number of strong blocks is necessarily
  constrained by the number of vertices of the K\H{o}nig graph, i.e., by the
  number of species and reactions, while a single strong block may contain
  many distinct autocatalytic cores. In view also of our results in
  \cite{golnik_enumeration_2026,golnik_bridging_2026}, this analysis rather
  suggests that it is the internal connectivity within a single strong
  block that provides the structural setting in which many distinct
  autocatalytic cores can occur.
\end{itemize}
Finally, and more informally, the Binary Polymer Model naturally allows a
large number of distinct species to act as catalysts. Such diversity seems
to provide the connectivity richness that gives rise to a large number of
autocatalytic cores within a single irreducible RAF.

In the next two abstract examples, we build upon the above intuitions. In the
first, both catalytic diversity and multiple catalyzations are combined to
produce an exponential number of autocatalytic cores within a single
irreducible RAF. In the second, we consider the opposite extreme, where all
F-reactions are monocatalyzed and share the same catalyst, and construct
arbitrarily many irreducible RAFs that possess only a single shared
autocatalytic core.

\begin{figure}
	\centering
	\begin{minipage}[c]{0.5\textwidth}
		\centering
		\includegraphics[width=\textwidth]{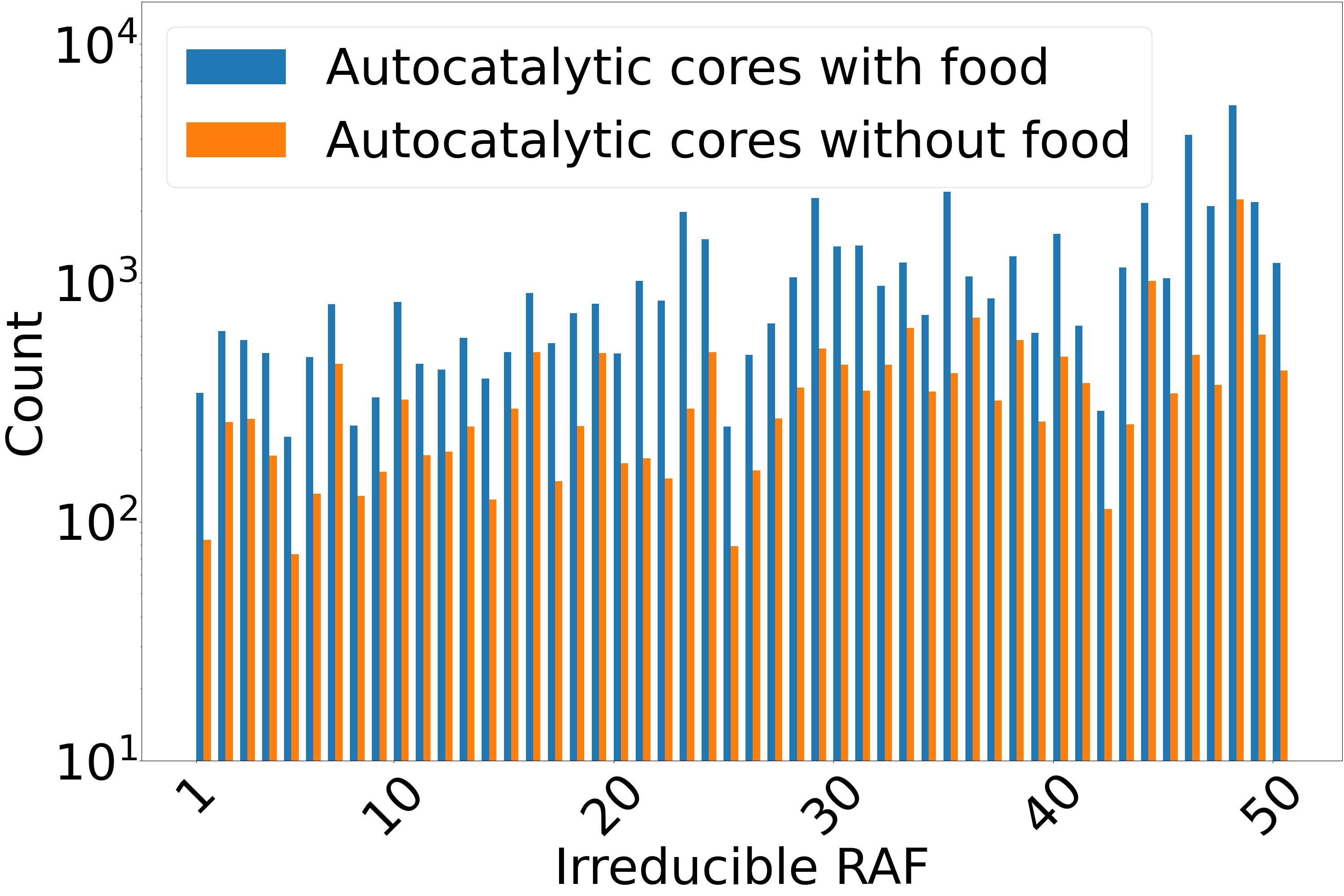}
	\end{minipage}%
	\hfill%
	\begin{minipage}[c]{0.5\textwidth}
		\centering
		\includegraphics[width=\textwidth]{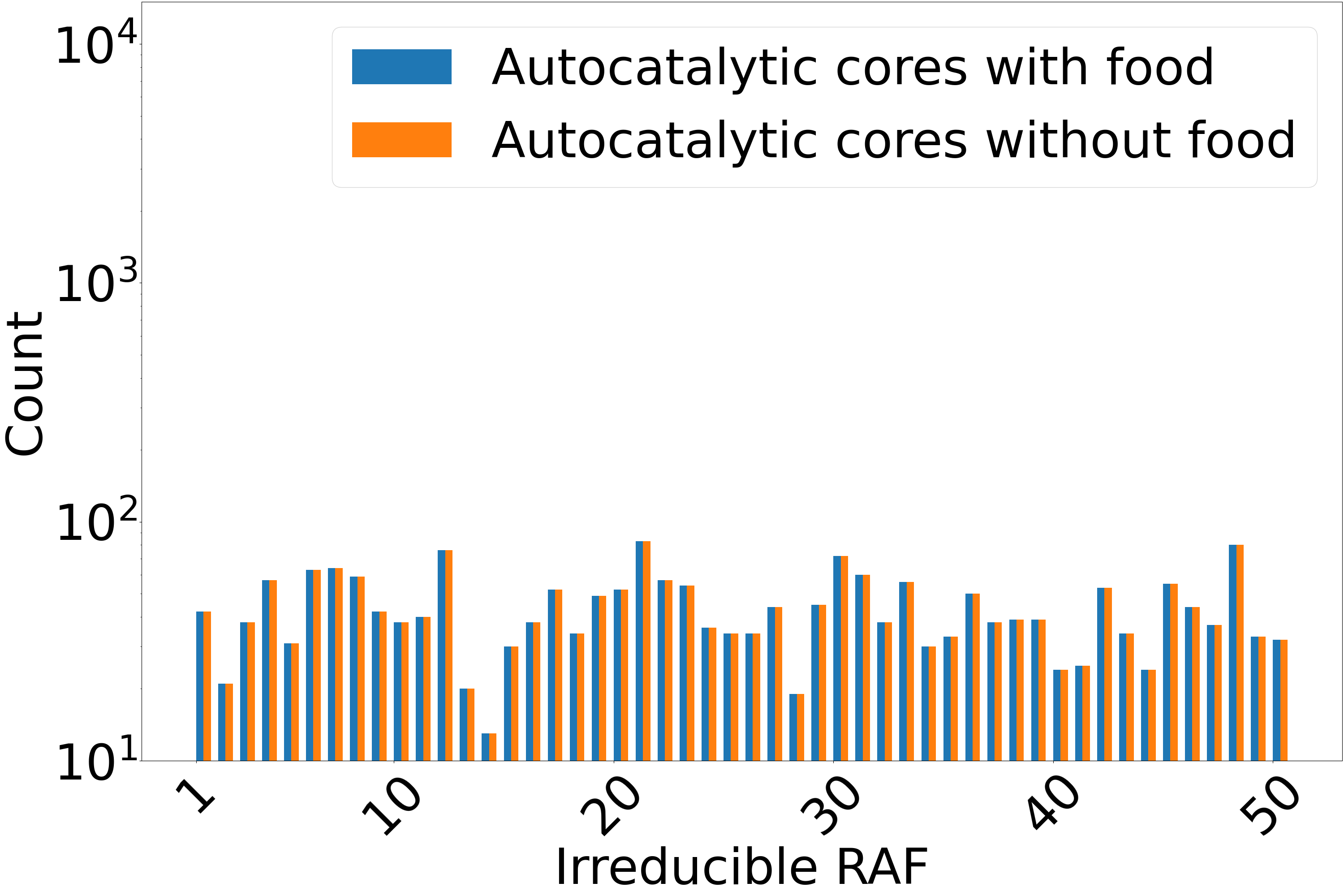}
	\end{minipage}
	\caption{Autocatalytic cores in 50 different arbitrarily chosen irreducible RAFs from two independent 
	instances of the BPM depicted and described in Fig.~\ref{fig:irrRAFCores} and 
	Sect.~\ref{sec:BinPol}) with and without food molecules.}
	\label{fig:FoodVsNonFood} 
\end{figure}

\subsection{The necklace RAF}\label{sec:NecklaceRAF}

\begin{figure}
		\centering
		\begin{tikzpicture}
			\node[draw, circle] (f) at (1,-1.5) {$f$};
			\node[draw, circle] (x'1) at (0,1) {$x'_1$};
			\node[draw, circle] (x''1) at (2,1) {$x''_1$};
			\node[draw, circle] (x2) at (1,-5) {$x_2$};
			\node[draw, circle] (x'2) at (6,-5) {$x'_2$};
			\node[draw, circle] (x''2) at (6,-3) {$x''_2$};
			\node[draw, circle] (x3) at (8,-5) {$x_3$};
			\node[draw, circle] (x'3) at (5,1) {$x'_3$};
			\node[draw, circle] (x''3) at (5,-1) {$x''_3$};
			
			\node[draw, rectangle] (r'1) at (0,-2) {$r'_1$};
			\node[draw, rectangle] (r''1) at (2,-2) {$r''_1$};
			\node[draw, rectangle] (r'2) at (4.5,-6) {$r'_2$};
			\node[draw, rectangle] (r''2) at (4.5,-4) {$r''_2$};			
			\node[draw, rectangle] (r'3) at (7,2) {$r'_3$};
			\node[draw, rectangle] (r''3) at (7,0) {$r''_3$};

			
			
			\draw[->] (f) to (r'1);
			\draw[->] (f) to (r''1);
			\draw[->, dashed] (x'1) to [bend right = 30](r'1);
			\draw[->, dashed] (r'1) to [bend right = 30] (x'1);
			\draw[->, dashed] (x''1) to [bend right = 30](r''1);
			\draw[->, dashed] (r''1) to [bend right = 30] (x''1);
			
			\draw[->] (r'1) to [bend right = 30] (x2);
			\draw[->] (r''1) to [bend left = 30]  (x2);
			\draw[->] (x2) to [bend right = 15] (r'2);
			\draw[->] (x2) to [bend left = 15] (r''2);			
			\draw[->, dashed] (x'2) to [bend right = 15] (r'2);
			\draw[->, dashed] (r'2) to [bend right = 15] (x'2);
			\draw[->, dashed] (x''2) to [bend right = 15] (r''2);
			\draw[->, dashed] (r''2) to [bend right = 15] (x''2);			
			\draw[->] (r'2) to [bend right = 15] (x3);
			\draw[->] (r''2) to [bend left = 15] (x3);
			
			\draw[->] (x3) to [bend right=30] (r'3);
			\draw[->] (x3) to [bend left=15] (r''3);
			\draw[->] (r'3) to [bend right=15] (x'1);
			\draw[->] (r'3) to [bend right=15] (x''1);
			\draw[->] (r''3) to [bend left=15] (x''1);
			\draw[->] (r''3) to [bend left=15] (x'1);			
			
			\draw[->, dashed] (x'3) to [bend left=15] (r'3);
			\draw[->, dashed] (r'3) to [bend left=15] node[midway, below] {$2$} (x'3);
			\draw[->, dashed] (x''3) to [bend left=15] (r''3);
			\draw[->, dashed] (r''3) to [bend left=15] node[midway, below] {$2$} (x''3);
						
												
			
%
%
%
											
		\end{tikzpicture}		
	\caption{The K\H{o}nig graph of the 2-necklace RAF with $n=3$. 
	The edges $(r'_3, x'_2), (r'_3, x''_2), (r''_3, x'_2)$, and $(r'_3, x''_2)$ were omitted for increased visibility. 
	Edges connecting catalysts with reaction vertices are depicted with dashed lines.}
	\label{fig:NecklaceRAF}
\end{figure}
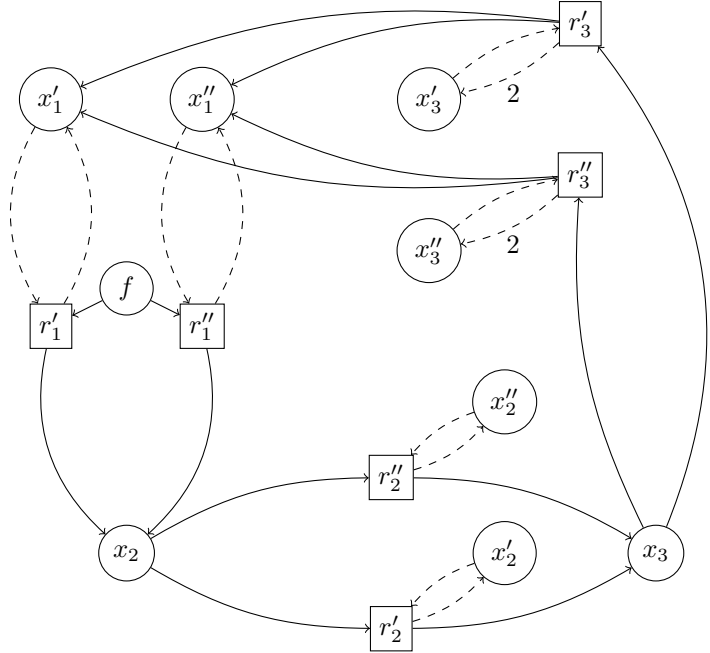

We present an example of an irreducible RAF with $3n$ species and $n$
F-reactions, each of which possesses two distinct catalyzations and no
species catalyzes more than one F-reaction. The RAF contains an exponential number (w.r.t. the number of F-reactions) of autocatalytic cores. We use the CRN notation for simplicity, and we
highlight the two reactions in the same F-reaction $\fr_i$, $i=1,\dots,n$,
as $r'_i$ and $r''_i$. We assume throughout that $n\geq 3$.
\begin{equation*}
  \begin{cases}
    f + (x'_1) &
    \underset{r'_1}{\longrightarrow} x_2 + (x'_1)  \\
    f + (x''_1) &
    \underset{r''_1}{\longrightarrow} x_2 + (x''_1)  \\
    x_2 + (x'_2) &
    \underset{r'_2}{\longrightarrow} x_3 + (x'_2)  \\
    x_2 + (x''_2) &
    \underset{r''_2}{\longrightarrow} x_3 + (x''_2)  \\
    &\vdots\\
    x_k + (x'_k) &
    \underset{r'_k}{\longrightarrow}
    x_{k+1} + x'_{k-1} + x''_{k-1} + (x'_k) \qquad\qquad k=3,\ldots,n-1. \\
    x_k + (x''_k) &
    \underset{r''_k}{\longrightarrow}
    x_{k+1} + x'_{k-1} + x''_{k-1} + (x''_k)  \\
    &\vdots\\
    x_n + (x'_n) &
    \underset{r'_n}{\longrightarrow}
    x'_1 + x''_1 + x'_{n-1} + x''_{n-1}
    + x'_n + x''_n + (x'_n)  \\
    x_n + (x''_n) &
    \underset{r''_n}{\longrightarrow}
    x'_1 + x''_1 + x'_{n-1} + x''_{n-1}
    + x'_n + x''_n + (x''_n)
  \end{cases}
\end{equation*}
In view of its geometry and architecture, we call the example the
$2$-\emph{necklace RAF} of size $n$. See
Fig.~\ref{fig:NecklaceRAF} for a depiction of the case $n=3$.

The two reactions $r'_i$ and $r''_i$ have identical F-reactants and
F-products and differ only in their catalyzation, and thus correspond to
the same F-reaction $\fr_i$. The food set is the singleton $\food=\{f\}$ and it is easy to see that $\RAF=\{\fr_1,\ldots,\fr_n\}$
is an irreducible RAF, as both the F-catalysts of the first $\fr_1$ are produced only via the last F-reaction $\fr_n$ of the chain. 
The same observation shows that the K\H{o}nig graph associated to the entire network possesses only one maximal strong 
block containing non-food species. Each $2$-\emph{necklace RAF} has exactly $\mathbf{2}^n+\mathbf{2}^2(n-2)+\mathbf{2}^1$ 
autocatalytic cores identified in three groups, corresponding to the
three summands:
\begin{itemize}
\item The CRS $\RAF$ contains exactly $2^n$ distinct monocatalyzed
sub-CRSs, each of which identifies an autocatalytic core of size $n$,
containing all the $n$ reactions denoted generically as $(r_1,r_2,...,r_n)$ and $n$ species $(\cat(r_1),x_2,...,x_n)$. The associated stoichiometric matrix has the identical form
\[S_n=
  \begin{pmatrix}
    0 & 0 & 0 & \cdots & 0 & 1\\
    1 & -1 & 0 & \cdots & 0 & 0\\
    0 & 1 & -1 & \cdots & 0 & 0\\
    \vdots & \vdots & \vdots & \ddots & \vdots & \vdots\\
    0 & 0 & \cdots & 1 & -1 & 0\\
    0 & 0 & \cdots & 0 & 1 & -1
  \end{pmatrix}.
\]
The subnetwork is naturally well-formed, and (via the construction in Lemma \ref{lem:NonFoodSemipositiv2}) choosing $v=(n,n-1,\ldots,2,1)^T$ shows semipositivity since $
  S_nv=(1,\ldots,1)^\top>0$. The minimality also follows directly from the cycle structure of the autocatalytic core, as removing any set of columns destroys semipositivity and 
  removing any set of rows destroys well-formedness. As different monocatalyzation choices identify different CRN reactions, the autocatalytic cores are distinct and of precise 
  number $2^n$.
\item Each pair $(\fr_k,\fr_{k+1})$ of consecutive F-reactions, with $k=2,\ldots,n-1$, identifies each $2^2=4$ distinct
autocatalytic cores of size $2$. The associated species are chosen $(\cat(r_k),x_{k+1})$, and the associated stoichiometric matrix reads
\[ \begin{pmatrix}
    0 & 1\\
    1 & -1
  \end{pmatrix}.
\]
There are $n-2$ such consecutive pairs and $2\cdot2=4$ choices of catalyzations 
for each pair, yielding $4(n-2)$ distinct autocatalytic cores of size $2$. 
\item Different choices of catalyzation for the F-reaction $\fr_n$,
i.e. reactions $r'_n$ and $r''_n$, identify two additional
one-dimensional autocatalytic cores. Indeed,
\[
  S[\{x'_n\},\{r'_n\}]=(1),
  \qquad
  S[\{x''_n\},\{r''_n\}]=(1),
\]
since $r'_n$ and $r''_n$ produce one additional copy of their
respective F-catalyst. Hence each reaction alone identifies an
autocatalytic core.

\end{itemize}
It can be verified that these three families exhaust all autocatalytic cores of the
$2$-\emph{necklace RAF}. We omit a detailed argument since a lower exponential bound already serves our point. Via the arguments above and $
  1+(n-2)+1=n$, we also remark that each monocatalyzed sub-CRS of $\RAF$ contains $n$ 
autocatalytic cores: one of size $n$, $n-2$ of size $2$, and one of
size $1$. This shows indeed that a high number of cores is achievable also for monocatalyzed networks. Note as well that different monocatalyzed sub-CRSs share many of the
size-$2$ and size-$1$ cores.

The construction and calculation can be generalized to an
$m$-\emph{necklace RAF} of size $n$, where each F-reaction has $m$
distinct catalyzations. The same argument gives
  $m^n+m^2(n-2)+m$
autocatalytic cores. More generally, suppose that the F-reaction $\fr_i$ has $m_i$ distinct
catalyzations, with the same necklace architecture as above. The number
of distinct autocatalytic cores is then
\begin{equation}
  \prod_{i=1}^{n}m_i
  +
  \sum_{k=2}^{n-1}m_km_{k+1}
  +
  m_n.
\end{equation}
We stress that such construction heavily relies indeed on the fact that each
species is a catalyst of at most one reaction. In contrast, when this is
not the case, we emphasize that catalyzations might indeed affect
the minimality property of autocatalytic sub-CRNs. 
The following example (see also Fig.~\ref{fig:Minimality}) illustrates a situation where two
autocatalytic sub-CRNs differ only by catalyzation;
\begin{equation}\label{eq:MinCat}
  \begin{cases}
    f + (x_1) &
    \underset{r_1}{\rightarrow} x_2 + x_4 + (x_1) \\[0.5em]
    x_2 + (x_1) &
    \underset{r'_2}{\rightarrow} x_2 + x_3 + (x_1) \\[0.5em]
    x_2 + (x_4) &
    \underset{r''_2}{\rightarrow} x_2 + x_3 + (x_4) \\[0.5em]
    x_3 + (x_3) &
    \underset{r_3}{\rightarrow} x_1 + (x_3)
  \end{cases}
\end{equation}
in particular, $
  (\{x_1,x_2,x_3\},\{r_1,r''_2,r_3\})$
is a core, while 
$
  (\{x_1,x_2,x_3\},\{r_1,r'_2,r_3\})
$
is not. In the latter, the catalyzation of $r'_2$ by
$x_1$ creates the proper autocatalytic sub-CRN
\[
  (\{x_1,x_3\},\{r'_2,r_3\}),
\]
so that the three-reaction sub-CRN is not minimal. In contrast, replacing
$r'_2$ by $r''_2$ removes this smaller feedback structure because its
F-catalyst $x_4$ is not contained in the considered species set, and the
three-reaction sub-CRN is an autocatalytic core. In the next example we 
present a construction at the opposite extreme: a RAF where all reactions 
are catalyzed by one single species and indeed it admits only one autocatalytic 
core but multiple irreducible RAFs.

\begin{figure}[htb]
	\centering
		\begin{tikzpicture}
			\node[draw, circle] (f) at (1.5,1.5) {$f$};		
			\node[draw, circle] (x1) at (0,0) {$x_1$};
			\node[draw, circle] (x2) at (3,0) {$x_2$};
			\node[draw, circle] (x3) at (3,-2) {$x_3$};
			\node[draw, circle] (x4) at (6,-1) {$x_4$};	
			
			\node[draw, rectangle] (r1) at (1.5,0) {$r_1$};
			\node[draw, rectangle] (r2') at (1.5,-1) {$r'_2$};
			\node[draw, rectangle] (r2'') at (4.5,-1) {$r''_2$};
			\node[draw, rectangle] (r3) at (0,-2) {$r_3$};
			
			\draw[->] (f) to (r1);
			\draw[->, dashed] (x1) to [bend right = 15] (r1);
			\draw[->, dashed] (r1) to [bend right = 15] (x1);			
			\draw[->] (r1) to (x2);
			\draw[->] (r1) to [bend left = 60, looseness = 1.5] (x4);
			
			\draw[->] (x2) to (r2');
			\draw[->] (x2) to (r2'');
			
			\draw[->, dashed] (x1) to [bend right = 10] (r2');
			\draw[->, dashed] (r2') to [bend right = 10] (x1);			
			\draw[->, dashed] (x4) to [bend right = 30] (r2'');
			\draw[->, dashed] (r2'') to [bend right = 30] (x4);
			
			\draw[->] (r2') to (x3);
			\draw[->] (r2'') to (x3);
			
			\draw[->, dashed] (x3) to [bend right = 15] (r3);
			\draw[->, dashed] (r3) to [bend right = 15] (x3);
			
			\draw[->] (r3) to (x1);

		\end{tikzpicture}
	\caption{The bipartite K\"{o}nig graph of the irreducible RAF of Eq.~\eqref{eq:MinCat}, 
	in which two catalyzations for one F-reaction influence the minimality of an autocatalytic 
	subsystem and give rise to two distinct autocatalytic cores of different size.}
	\label{fig:Minimality}
\end{figure}
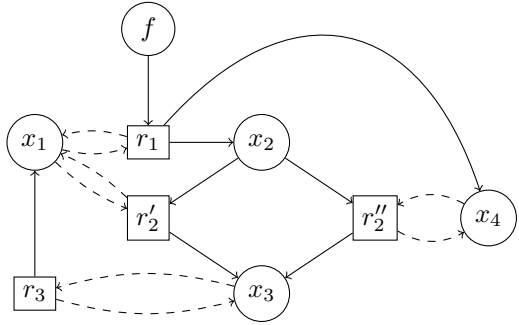

\subsection{Many irreducible RAFs and just one autocatalytic core}\label{sec:1CMiRAF}

Up to now, the presented examples suggest how the number of autocatalytic
cores can be very large in a single irreducible RAF.  However, this
observation really depends on the network architecture. In fact, we next
present an example of a family of networks that may contain any number $n$
of irreducible RAFs and invariably only one autocatalytic core. Since such
a family is monocatalyzed, we use the CRN notation for simplicity.

Consider a simple autocatalytic reaction $r_A$:
\[x_S+(x_K)\quad\underset{r_A}{\longrightarrow}\quad x_K + (x_K), \]
which is stoichiometrically equivalent to the standard autocatalytic reaction $x_S+x_K\rightarrow 2x_K$. Assume now $x_S$ is not 
in the food set, but it can be produced from $n$ distinct food molecules via distinct reactions, all catalyzed by $x_K$, i.e.
\begin{equation}\label{eq:OneCoreMulIrrRAF}
    \begin{cases}
        f_1+(x_K)\quad&\underset{r_1}{\longrightarrow}\quad x_S + (x_K)\\
         f_2+(x_K)\quad&\underset{r_2}{\longrightarrow}\quad x_S + (x_K)\\
         &\vdots \\
         f_n+(x_K)\quad&\underset{r_n}{\longrightarrow}\quad x_S + (x_K)
    \end{cases}
\end{equation}
Consider now the network consisting of the $n+2$ species $X=(f_1,...,f_n,x_S,x_K)$ of which the food set is defined as $\food=\{f_1,...,f_n)$, 
together with the $n+1$ reactions $\mathcal{R}=(r_1,..., r_N, r_A)$. The network (see Fig.~\ref{fig:OneCoreMulIrrRAF}) contains $n$ irreducible 
RAFs of the form $\mathcal{R}'_i=\{r_i,r_A\}$ for $i=1,\ldots,n$. Each of those contains the same autocatalytic core, identified by $r_A$ alone, which is 
stoichiometrically-one dimensional and corresponds to 
\[S[\{x_K\},\{r_A\}]=(1).\]
A similar architecture was already presented in \cite{golnik_bridging_2026}, Example 5.5. The structural ingredients can be summarized as:
\begin{enumerate}
    \item \emph{Few species play the role of F-catalysts}: all reactions share the same F-catalyst;
    \item \emph{No autocatalytic cores involving food molecules}: food molecules only appear as F-reactants but not as F-products.
\end{enumerate}

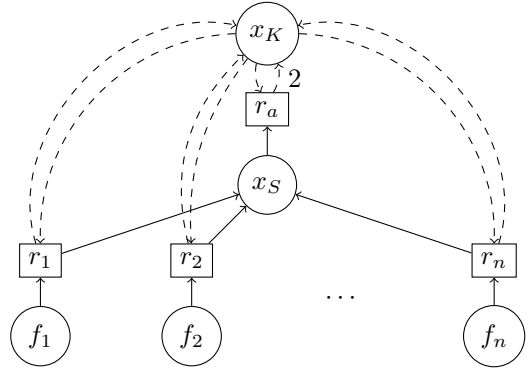
\begin{figure}[htb]
	\centering
	\begin{tikzpicture}
		\node[draw, circle] (f1) at (0,0) {$f_1$};
		\node[draw, circle] (f2) at (2,0) {$f_2$};
		\node[] 		    (d) at (4,0.5) {$\cdots$};
		\node[draw, circle] (fn) at (6,0) {$f_n$};
		\node[draw, circle] (xs) at (3, 2) {$x_S$};
		\node[draw, circle] (xk) at (3, 4) {$x_K$};
		
		\node[draw, rectangle] (r1) at (0, 1) {$r_1$};
		\node[draw, rectangle] (r2) at (2, 1) {$r_2$};
		\node[draw, rectangle] (rn) at (6, 1) {$r_n$};
		\node[draw, rectangle] (ra) at (3, 3) {$r_a$};
		
		\draw[->] (f1) to (r1); 
		\draw[->] (f2) to (r2);
		\draw[->] (fn) to (rn);
		
		\draw[->] (r1) to (xs);
		\draw[->] (r2) to (xs); 
		\draw[->] (rn) to (xs);
		
		\draw[->] (xs) to (ra);
		\draw[->, dashed] (ra) to [bend right = 20] node[midway, right] {$2$} (xk);
		\draw[->, dashed] (xk) to [bend right = 20] (ra);
		
		\draw[->, dashed] (r1) to [bend left = 60] (xk);
		\draw[->, dashed] (xk) to [bend right = 45] (r1);
		\draw[->, dashed] (r2) to [bend left = 30] (xk);
		\draw[->, dashed] (xk) to [bend right = 20] (r2);
		\draw[->, dashed] (rn) to [bend right = 60] (xk);
		\draw[->, dashed] (xk) to [bend left = 45] (rn);
	\end{tikzpicture}
	\caption{Depiction of the bipartite K\H{o}nig graph of the irreducible specified in Eq.~\eqref{eq:OneCoreMulIrrRAF}.}
	\label{fig:OneCoreMulIrrRAF}
\end{figure}

\section{Discussion}

\textsc{Reflexively autocatalytic sets (RAFs)} and \emph{stoichiometric
  autocatalysis} are grounded in different formalizations of chemical
reaction systems that lead to inequivalent concepts of reactions,
distinguished by the manner in which catalysis is handled. While in CRNs,
underlying the theory of \emph{stoichiometric autocatalysis}, catalysts are
not formally distinguished from other reactants and products, the CRS 
framework distinguishes F-reactions, intended as molecular transformations, 
from their catalyzations, which are specified separately. As a
consequence, F-reactions in CRS are equivalence classes of CRN reactions with
respect to the equivalence relation $\sim$ (Eq.~\eqref{eq:EqRel}),
characterizing reactions that share the same non-catalyst reactants
(F-reactants) and the same non-catalyst products (F-products), with
identical stoichiometries, see Sec.~\ref{sec:CRNCRS}.

From a set of CRN reactions $\CR$, the set of corresponding F-reactions
$\FR$ can be obtained as the quotient set of $\CR$, i.e.
\begin{equation}
  \FR=\CR/\mathord{\sim}.
\end{equation} 

This implies that the CRS $\mathbf{\Phi}(\Gamma)= (X, \CR/\mathord{\sim},
\kappa_F(\CR/\mathord{\sim}))$, with $\kappa_F(\CR/\mathord{\sim})$
determined as described in Eq.~\eqref{eq:CatCFR}, reflects a quotient
structure of the CRN $\Gamma=(X, \CR)$. In turn, the reactions of the CRN
$\Gamma(\mathbf{\Phi})$ associated with a given CRS 
$\mathbf{\Phi}\coloneqq (X, \FR, \kappa_F(\FR))$ can be obtained from
$\kappa_F(\FR)$, since each catalyzation of an F-reaction specifies a
unique CRN reaction.

It was recently demonstrated that every \textsc{RAF} that is itself not a \textsc{CAF} 
is \emph{stoichiometrically autocatalytic} \cite{golnik_bridging_2026}. Here,
we focused on the interplay between the minimal structures arising in both frameworks,
i.e., \textsc{irreducible RAFs} and \emph{autocatalytic cores}. Analyzing 
irreducible RAFs obtained from random instances of the Binary Polymer Model
\cite{kauffman_autocatalytic_1986} with \texttt{autogatito} \cite{golnik_using_2026}, 
we observed a large number of autocatalytic cores in each 
irreducible RAF (Fig.~\ref{fig:irrRAFCores}). This observation could not 
be explained by the theoretical results stated in \cite{golnik_bridging_2026}.

Irreducible RAFs are sets of F-reactions minimal w.r.t. the RAF properties 
of Def.~\ref{def:RAF}. Consequently, among two reactions sharing the same set of 
F-reactants and F-products, at most one can be an element of an irreducible RAF 
(Lemma~\ref{lem:OneFReacPerRAF}). In the same vein, two F-reactions $\fr_1$ 
and $\fr_2$ that satisfy $\feducts(\fr_1)=\fproducts(\fr_2)$ and vice versa cannot  
both be part of the same irreducible RAF. In particular, an F-reaction and its reverse 
cannot both belong to the same irreducible RAF and reversing F-reactions within 
an irreducible RAF cannot produce another irreducible RAF (Lemma~\ref{lem:noReverseIrrRAF}).

The translation from F- to CRN-reactions suggests introducing 
an additional layer of minimality. In particular, Lemma~\ref{lem:mono} 
guarantees that an irreducible RAF $\RAF$ retains the minimality property in 
any CRS, where each F-reaction of $\RAF$ is associated with a single 
catalyzation. In this \emph{monocatalyzed} case, CRN and F-reactions 
are in one-to-one correspondence; each equivalence class consists of 
a single element, and thus CRN and CRS can be canonically identified 
with each other. For this reason, we focused on monocatalyzed, irreducible 
RAFs $\RAF$. This allowed us to investigate their properties by examining 
the structure of the bipartite K\H{o}nig graph $\king(\RAF)$ of the corresponding 
CRN representation.

RAF-irreducibility endows $\king(\RAF)$ with three fundamental properties for
any irreducible, monocatalyzed RAF $\RAF$
(Lemmata~\ref{lem:PSC}--\ref{lem:StrongBlocks}).
\begin{itemize}
	\item[I1.] The subgraph of $\king(\RAF)$ obtained by deleting food and waste species vertices is strongly connected.
	\item[I2.] Every cut vertex of $\king(\RAF)$ is a reaction vertex.
	\item[I3.] Every maximal strong block $\block$ of $\king(\RAF)$ that contains 
		      at least one non-food, non-waste species vertex contains a reaction 
		      $r^*$, also referred to as \emph{entry reaction}, that exerts the function of 
		      an inflow reaction for $\block$ -- that is, all non-food F-reactants of $r^*$ are located 
		      outside $\block$, at least one F-product and one F-catalyst that is not a food-species 
		      is located inside $\block$.
\end{itemize}
Together, these properties imply that every maximal strong block containing at least one non-food, non-waste species vertex 
is \emph{stoichiometrically autocatalytic} and therefore contains in particular an autocatalytic core. Indeed, by definition each 
strong block is strongly connected, which guarantees well-formedness, while the entry reaction identified in I3 allowed us to 
establish semipositivity of the corresponding stoichiometric submatrix (Lemma~\ref{lem:NonFoodSemipositiv2}).

In particular, our results show that the number of autocatalytic cores in an irreducible RAF is at least as large as the number 
of relevant maximal strong blocks in any of its monocatalyzed realizations. However, different monocatalyzed realizations may 
share autocatalytic cores. More importantly, the analysis of the BPM instances suggests that the internal connectivity of a single 
strong block, together with the availability of many distinct catalyst species and multiple catalyzations of the same F-reaction, 
may play a larger role in the abundance of autocatalytic cores. Indeed, our structural results, in particular Lemma~\ref{cor:strongblock1}, 
much constrain the occurrence of multiple strong blocks; for instance, they are excluded in the ligation-only case (Cor.~\ref{cor:strongblock2}). 
All irreducible RAFs from the BPM analyzed here indeed contained only one highly connected strong block (after excluding food and waste species). 
Together with the occurrence of a few F-reactions possessing more than one catalyzation, this suggests that the observed 
multiplicity of cores arises mainly from the internal architecture of a single strongly connected block. To reinforce this intuition, 
Sec.~\ref{sec:NecklaceRAF} provided the architecture of a \emph{necklace RAF}, in which a single strong block contains an 
exponentially large number of autocatalytic cores. To stress that this relationship depends strongly on the network architecture, 
Sec.~\ref{sec:1CMiRAF} also provided an opposite construction with arbitrarily many irreducible RAFs but only one autocatalytic core, shared by all of them.

The results obtained here demonstrate that there are meaningful connections 
between the CRS and CRN formalisms that translate into similarities 
between two different notions of autocatalysis, in particular between \textsc{irreducible RAFs} 
and \emph{autocatalytic cores}. Although our results yield insights as well as 
useful implications, they do not establish 1-to-1 correspondences between 
concepts in the two alternative formal frameworks. As the present results in 
essence characterize minimal RAFs in terms of CRNs, the converse, i.e., a 
characterization of autocatalytic cores in the language of CRS, 
remains an open research project.

\section{Declarations}

\subsection{Availability of Data and Materials}

The implementation, models, and all necessary data are available at https://github.com/hollyritch/minRAFs.

\subsection{Competing interests}

The authors declare that they have no competing interests. 

\subsection{Funding}

This work has been supported by the Novo Nordisk Foundation (grant NNF21OC0066551 `MATOMIC'), 
during a visit by Wim Hordijk in Leipzig. Research in the Stadler lab is {supported} by the BMBF (Germany)
through DAAD project 57616814 (SECAI, School of Embedded Composite AI), and jointly with SMWK 
(Saxony) through the \emph{Center for Scalable Data Analytics and Artificial Intelligence Dresden/Leipzig} (SCADS24B).
High-performance computing resources were provided by the German Network for Bioinformatics Infrastructure (de.NBI).

\subsection{Authors' contributions}

RG, TG, WH, PFS, and NV developed the mathematical theory. RG generated and analyzed the data on autocatalytic cores and their 
graph representations in all irreducible RAFs. WH generated all instances of the Binary Polymer Model and collected irreducible 
RAFs. All authors contributed to the conceptualization, methodology, and writing of the manuscript.

\bibliography{RAFCores}

\newpage

\setcounter{figure}{0}
\setcounter{section}{0}

\begin{center}
\huge{\textbf{Supplementary Information}}
\end{center}

\section{Glossary on Graph Theory}\label{sec:Glossary}

\paragraph{Graphs} A graph $G=(V,E)$ is a tuple composed of a set of vertices and 
a set of edges. Edges encode the adjacency of elements $u,v\in V$, i.e., $\{u,v\}\in E$ 
whenever $u$ is adjacent to $v$. We write $V(G)$ and $E(G)$ for the set of vertices 
and edges of a given graph $G$.
 
\paragraph{Directed Graphs} In a directed graph $G\coloneqq (V,E)$, edges are ordered 
pairs instead of unordered pairs, and $(u,v) \in E$ whenever there is an edge from $u$ to 
$v$. The underyling undirected graph of $G$ is the graph $G_{\circ} \coloneqq (V, E_{\circ})$ with
$E_{\circ}\coloneqq \{ \{u,v\} \ \vert \ (u,v) \in E(G) \}$.

\paragraph{Subgraphs} A subgraph of a (directed) graph $G$, denoted as $G'\subseteq G$ is a graph $G=(V', E')$ 
such that $V'\subseteq V$ and $E'\subseteq E$. A subgraph \emph{induced} by a set of vertices 
$V'\subseteq V$ is the graph $G[V']\coloneqq (V', F)$ with $F\coloneqq \{ \{u,v\} \ \vert \ u,v\in V', 
\{u,v\} \in E\}$ or in the directed case $F\coloneqq \{(u,v) \ \vert \ u,v\in V', (u,v) \in E\}$.

\paragraph{Paths and Directed Paths} A (directed) path $\mathcal{P}_{v_1v_k}$ in a (directed) graph $G$ is 
a sequence of pairwise distinct vertices $(v_1,\ldots, v_k)$ such that $\{v_i, v_{i+1}\}\in E(G) \  ((v_i, v_{i+1})\in E(G)), 
1\leq i\leq k-1$.

\paragraph{(Strongly) Connected (Sub-)Graphs} A subgraph $G'\coloneqq (V', E')$ of a (directed) graph 
$G=(V,E)$ is called (strongly) connected if for any two vertices $u,v\in V'$ there is a path 
$\mathcal{P}_{uv}$, which, in the directed case, implies that there is a path from $u$ to $v$ and 
from $v$ to $u$. If $G$ is directed, then $G'\subseteq G$ is called connected if $G'_\circ$ is connected. 
If $G'$ is not connected, we may also write that $G$ is disconnected. $V'$ is called a (strongly) connected 
component if $G[V']$ is (strongly) connected and maximal, i.e., if there is no $v\in V\setminus V': V'\cup \{v\}$ 
is (strongly) connected. If $V'=V$ then we say that $G$ is (strongly) connected. 

\paragraph{Cut vertex} For a graph $G\coloneqq (V, E)$ a vertex $v\in V$ is called a cut-vertex if $G[V\setminus \{v\}]$ 
is disconnected. Accordingly, in the directed case, then $v\in V$ is a cut-vertex if $G_\circ[V\setminus \{v\}]$ is 
disconnected. 

\paragraph{Strong blocks} For a digraph $G\coloneqq (V, E)$ a strongly connected subgraph $G'\subseteq G$ is 
called a strongly connected block or strong block if it does not contain a cut-vertex.

\paragraph{Elementary circuit} An elementary circuit of length $k$ in a digraph $G\coloneqq (V,E)$ is a sequence 
of vertices $(v_1,\ldots, v_k, v_1)$ such that $(v_1,\ldots, v_k)$ is a path, $(v_k, v_1) \in E$. 

\paragraph{Digon} A digon is an elementary circuit of length $2$.

\section{The Binary Polymer Model}\label{sec:BinPol}

The binary polymer model (BPM) is an abstract model of polymer chemistry that was introduced to argue about the probability of emergence of autocatalytic sets in random reaction networks. The set of molecules $X$ consists of all binary strings up to a maximal specified length $n$. In general, the reaction set $\FR$ consists of all possible \emph{ligation} reactions concatenating two bit strings together into a product as long as it does not exceed the maximum polymer length $n$, together with their reverse \emph{cleavage} reactions that decompose a bit string into two smaller ones. We denote the reverse of an F-reaction $\fr$ with $\bar{\fr}$. The food set $\food$ consists of all binary strings up to a given length $t$, here always $t=2$, i.e., the food set consists of the monomer and dimers. However, for some models $\FR$ is composed of ligation reactions only.

Catalysis is assigned randomly. In particular, there is a fixed probability $p$ that a given molecule (binary string) can catalyze a given F-reaction. For instances containing both ligation and cleavage, F-reactions that are the inverse of each other share the same catalysts. In other words, to create an instance of the BPM, for predetermined $n$ and $t$, each molecule-reaction pair $(x,\fr), x \in X, \fr \in \FR$ is included in the set of catalyzations $\cat(\FR)$ with probability $p$. More formally, an instance of the BPM forms a CRS $\mathbf{\Phi}\coloneqq (X, \FR, \cat(\FR))$, where:\\

\begin{itemize}
  \item {$X = \{0,1\}^{\leq n} \text{ and } \FR \coloneqq \FR \cup \bar{\FR}$}\\[-0.5em]
  \item {$\FR \coloneqq \{x_i + x_j \rightarrow x_i \cdot x_j : (x_i, x_j \in X) \wedge (|x_i \cdot x_j| \leq n)\}$}\\[-0.5em]
  \item {$\bar{\FR} \coloneqq \{\bar{\fr} \ \vert \ \fr \in \FR \}$} \\[-0.5em]
  \item {$\mathbb{P}[(x, \fr) \in \cat(\FR)] = p \hspace{100pt} x\in X, \fr \in \FR$}\\[-0.5em]
  \item {$\mathbb{P}[(x, \bar{\fr}) \in \cat(\FR)] = \begin{cases} 1 & \text{if } (x, \fr) \in \cat(\FR) \\ 0 & \text{otherwise} \end{cases} \quad x\in X, \bar{\fr}\in \bar{\FR}$}\\[-0.5em]
  \item {$\food = \{0,1\}^{\leq t}$}\\[-0.5em]
\end{itemize}

\subsection{Simulation Details}

Table~\ref{tab:SimulationDetails} specifies simulation details for different instances of the BPM.
In particular, all models contributed to the generation of Fig.~\ref{fig:StrongBlocks}, while only \textit{n6\_split} 
and \textit{n6\_forw} were employed for the generation of Fig.~\ref{fig:irrRAFCores} and Fig.~\ref{fig:FoodVsNonFood} 
in the main text.

\begin{table}[htb]
	\caption{Specification of simulation details for irreducible RAFs from different instances of the BPM.
	Maximal polymer size is specified by $n$, while the maximum size of food molecules is specified by $t$. The 
	uniform probability for a non-food molecule to catalyze an F-reaction is given by $p$, while food species are excluded 
	from being catalysts. F-reactions are in general considered irreversible.}
	\label{tab:SimulationDetails}
	\begin{tabular}{p{1.5cm}|p{1cm}|p{1.5cm}|p{1.5cm}|p{2.0cm}|p{1.5cm}|p{1cm}}
	Name & n & t & p & food-catalysts & Ligation & Cleavage \\
	\hline 
	n4\_split     		& 	4 	& 	1 	& 	0.018 	& 	no 	& 	yes 		& 	yes 		\\
	n5\_forw   			& 	5  	& 	2 	& 	0.014 	& 	no 	& 	yes 		& 	no 		\\
	n5irrev 			& 	5 	& 	2 	& 	0.015 	& 	no 	& 	yes 		& 	no 		\\
	n5irrev2    		& 	5 	& 	2 	& 	0.017 	& 	no 	& 	yes 		& 	no 		\\
	n6\_split			& 	6 	& 	2 	& 	0.0025	& 	no 	& 	yes 		& 	yes		\\
	n6\_split\_red 		& 	6 	& 	2 	&	0.002	& 	no 	& 	yes 		& 	yes 		\\
	n6\_forw			& 	6 	& 	2	&      0.006	&	no 	& 	yes 		& 	no 		\\
	n7\_split			&      7	&	2	&	0.0012	&	no	& 	yes		& 	yes 		\\
	n7\_split\_red		&	7	&	2	&	0.0009	&	no	&	yes 		& 	yes		\\
	n7\_forw          		&	7	& 	2	&	0.0025	&	no	&	yes 		&	no 		\\
	n7\_forw\_red		&	7	&	2	&	0.0022	&	no	&	yes 		&	no 		\\
	n8\_split			&	8	&	2	&	0.00055	&	no	&	yes 		&	yes 		\\
	n8\_forw			&	8	&	2	&	0.0012	&	no	&	yes 		&	no 		\\
	\end{tabular}
\end{table}

\subsection{Maximal Strong Blocks}
\begin{figure}[bth]
	\begin{minipage}[c]{\textwidth}
		\begin{minipage}[c]{0.5\textwidth}
			\centering
		\includegraphics[width=0.8\textwidth]{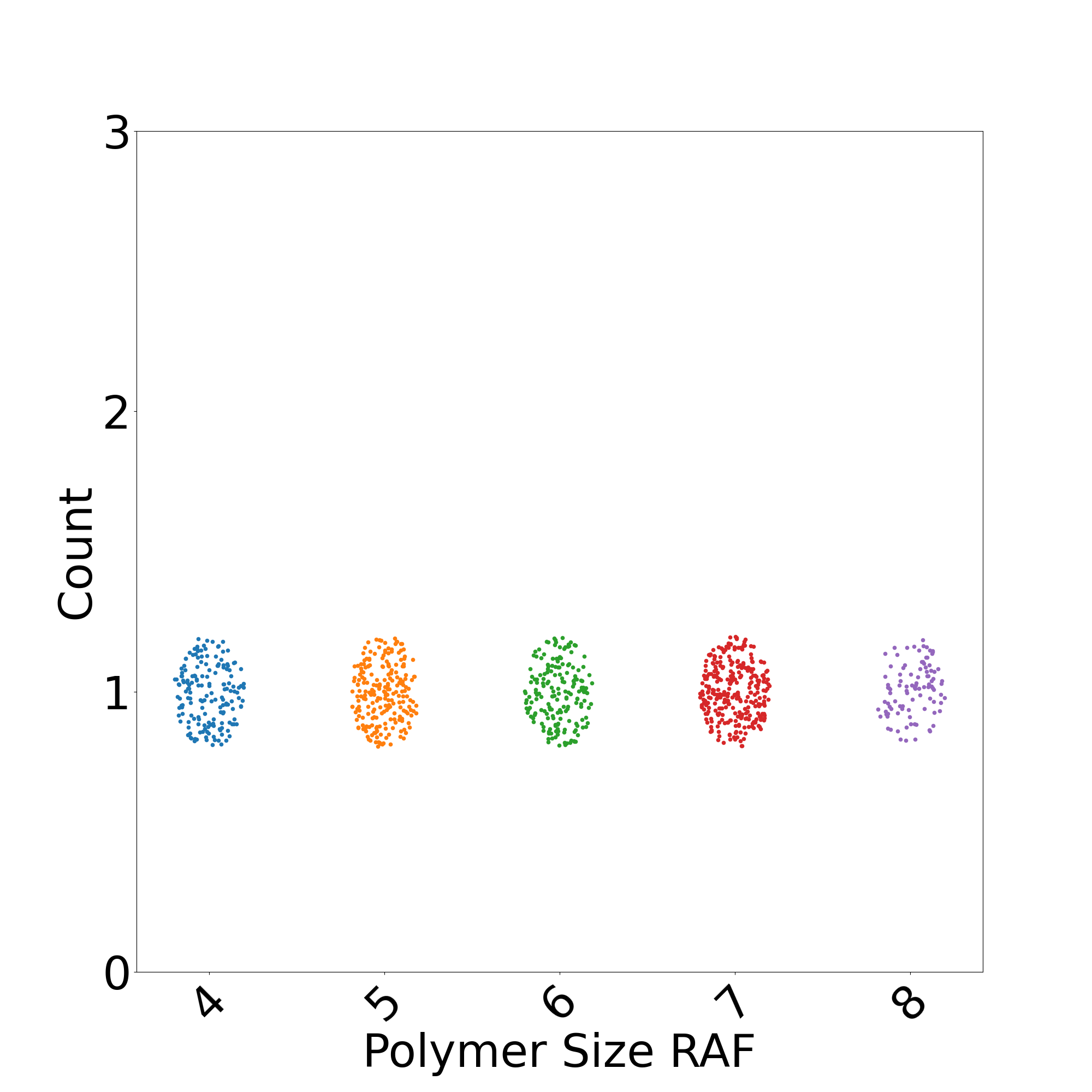}
		\end{minipage}%
		\hfill%
		\begin{minipage}[c]{0.5\textwidth}
			\centering
			\includegraphics[width=0.8\textwidth]{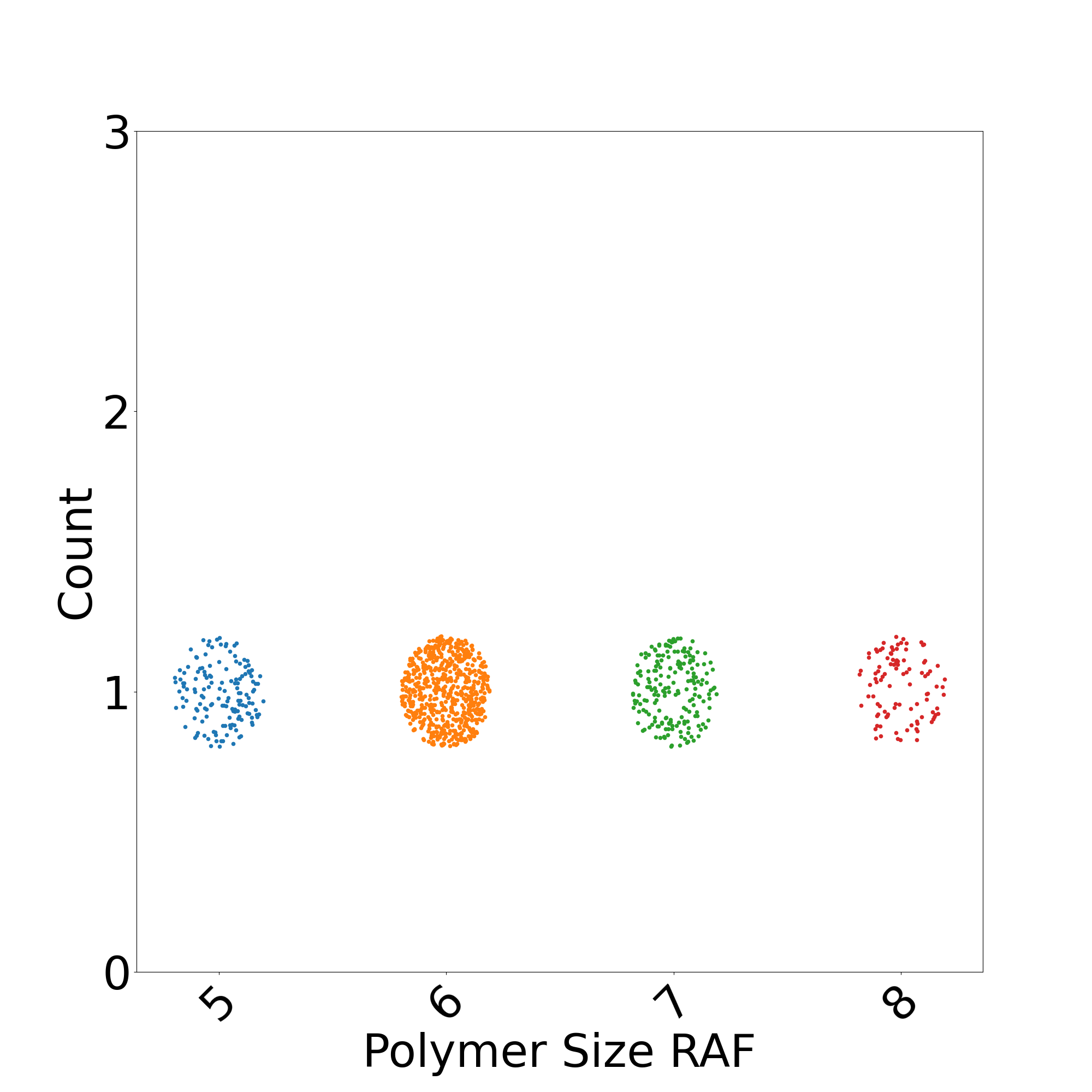}
		\end{minipage}
	\end{minipage}
	\begin{minipage}[c]{\textwidth}
		\begin{minipage}[c]{0.5\textwidth}
			\centering
		\includegraphics[width=0.8\textwidth]{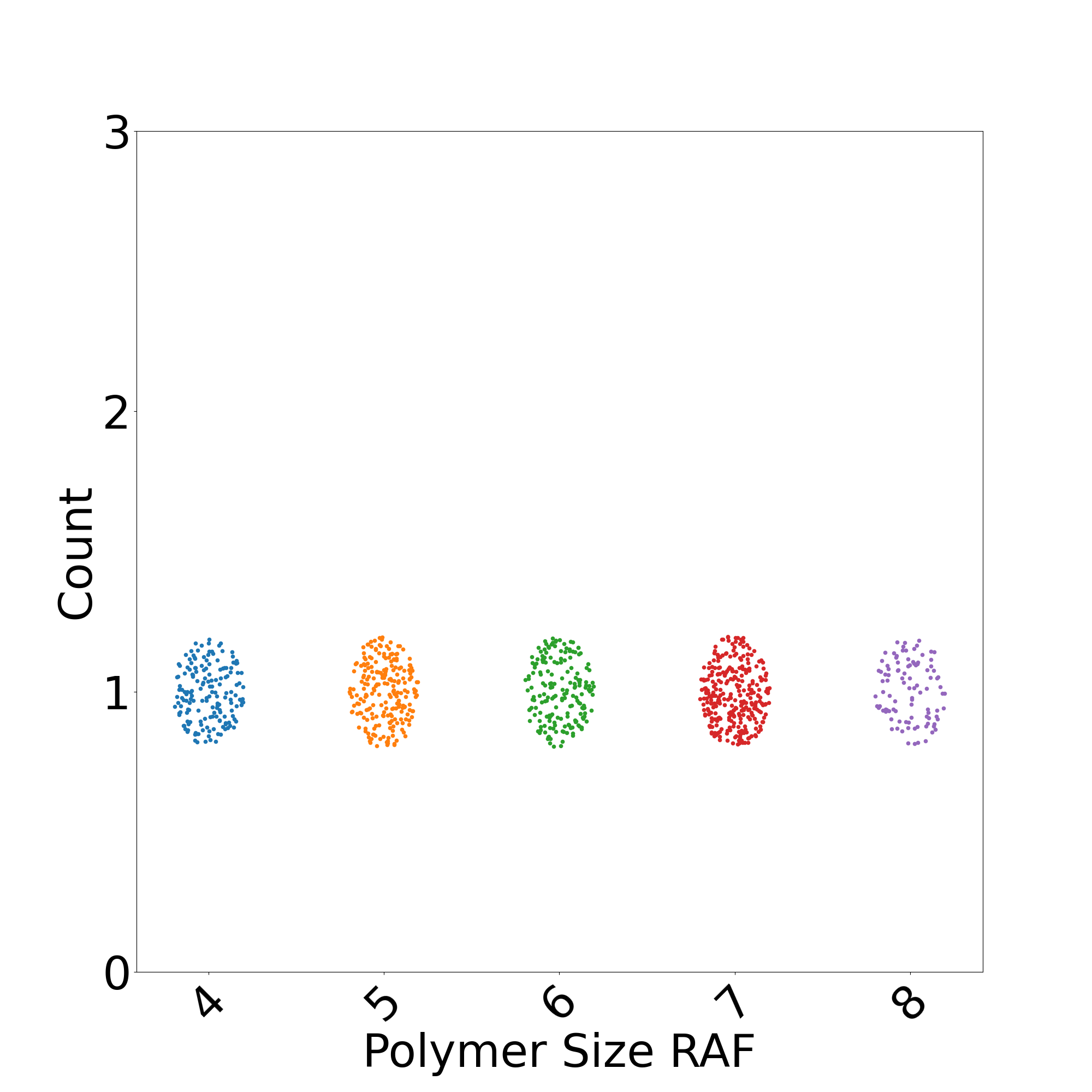}
		\end{minipage}%
		\hfill%
		\begin{minipage}[c]{0.5\textwidth}
			\centering
			\includegraphics[width=0.8\textwidth]{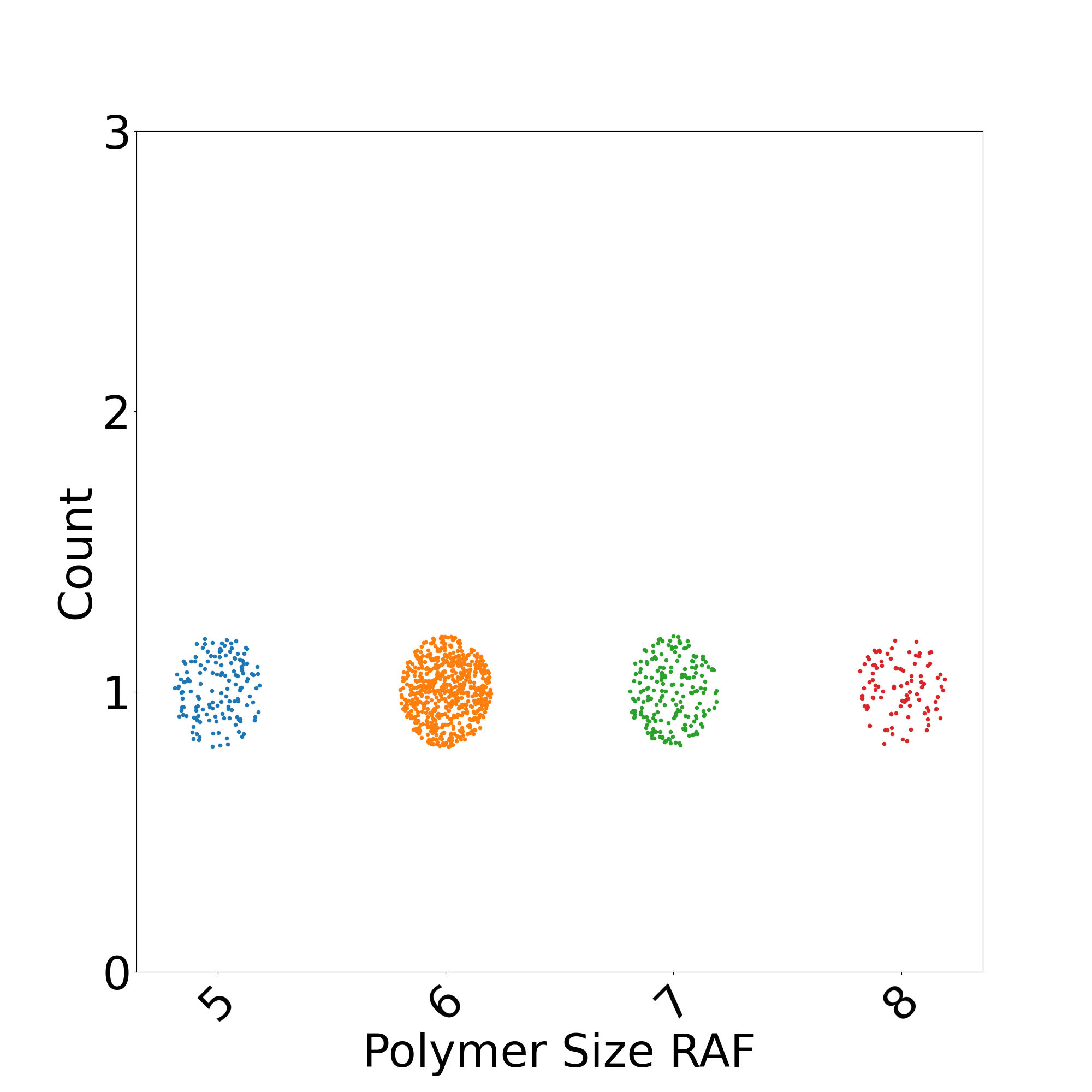}
		\end{minipage}
	\end{minipage}
		\caption{Number of maximal strong blocks containing at least one non-food, 
			non-waste species vertex in irreducible RAFs from different instances of the 
			BPM containing ligation/cleavage \textbf{(left)} or only ligation \textbf{(right)} 
			with \textbf{(upper)} and without \textbf{(lower)} food species. See SI 
			Table~\ref{tab:SimulationDetails}. Dots, each representing an irreducible 
			RAF with indicated maximal polymer length, are plotted in a non-overlapping manner.} 
		\label{fig:StrongBlocks}
\end{figure}

Figure~\ref{fig:StrongBlocks} shows the number of maximal strong blocks containing at least one non-food, 
non-waste species for all irreducible RAFs collected from the binary polymer models that are specified in 
Tab.~\ref{tab:SimulationDetails}.

\noindent Figure~\ref{fig:BPMMultipleStrongBlocks} depicts a theoretical example of an instance of the binary 
polymer model that is composed of two maximal strong blocks containing at least one non-food, 
non-waste species. However, as depicted in Fig.~\ref{fig:StrongBlocks}, such an instance was 
not observed by us. 

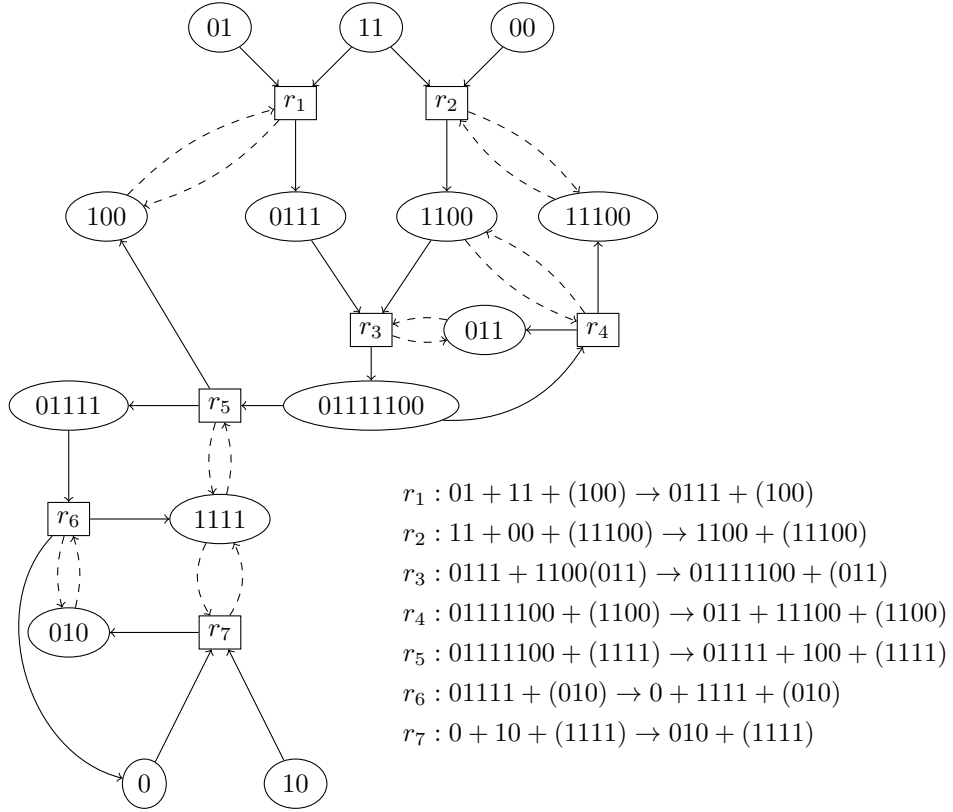
\begin{figure}[tbh]
	\begin{minipage}[c]{0.4\textwidth}
		\begin{tikzpicture}
			\node[draw, ellipse] (01) at (0,0) {$01$};
			\node[draw, ellipse] (11) at (2,0) {$11$};
			\node[draw, ellipse] (00) at (4,0) {$00$};						
			\node[draw, ellipse] (100) at (-1.5,-2.5) {$100$};
			\node[draw, ellipse] (0111) at (1,-2.5) {$0111$};
			\node[draw, ellipse] (011) at (3.5,-4) {$011$};
			\node[draw, ellipse] (1100) at (3,-2.5) {$1100$};
			\node[draw, ellipse] (11100) at (5,-2.5) {$11100$};
			\node[draw, ellipse] (01111100) at (2,-5) {$01111100$};
			\node[draw, ellipse] (01111) at (-2, -5) {$01111$};
			\node[draw, ellipse] (1111) at (0, -6.5) {$1111$};
			\node[draw, ellipse] (010) at (-2, -8) {$010$};
			\node[draw, ellipse] (0) at (-1, -10) {$0$};
			\node[draw, ellipse] (10) at (1, -10) {$10$};
			
			\node[draw, rectangle] (r1) at (1,-1) {$r_1$};
			\node[draw, rectangle] (r2) at (3,-1) {$r_2$};
			\node[draw, rectangle] (r3) at (2,-4) {$r_3$};
			\node[draw, rectangle] (r4) at (5,-4) {$r_4$};
			\node[draw, rectangle] (r5) at (0,-5) {$r_5$};
			\node[draw, rectangle] (r6) at (-2,-6.5) {$r_6$};
			\node[draw, rectangle] (r7) at (0,-8) {$r_7$};
			
			\draw[->] (01) to (r1);
			\draw[->] (11) to (r1);
			\draw[->] (r1) to (0111);
			\draw[->, dashed] (100) to [bend left = 15] (r1);
			\draw[->, dashed] (r1) to [bend left = 15] (100);

			\draw[->] (11) to (r2);
			\draw[->] (00) to (r2);			
			\draw[->] (r2) to (1100);
			\draw[->, dashed] (11100) to [bend left = 15] (r2);
			\draw[->, dashed] (r2) to [bend left = 15] (11100);
			
			\draw[->] (0111) to (r3);
			\draw[->] (1100) to (r3);
			\draw[->, dashed] (r3) to [bend right = 15] (011);
			\draw[->, dashed] (011) to [bend right = 15] (r3);
			\draw[->] (r3) to (01111100);
			
			\draw[->] (01111100) to [bend right = 30] (r4);
			\draw[->] (r4) to (11100);
			\draw[->] (r4) to (011);
			\draw[->, dashed] (r4) to [bend right = 15] (1100);
			\draw[->, dashed] (1100) to [bend right = 15] (r4);
			
			\draw[->] (01111100) to (r5);
			\draw[->] (r5) to (100);
			\draw[->] (r5) to (01111);
			\draw[->, dashed] (r5) to [bend right = 15] (1111);
			\draw[->, dashed] (1111) to [bend right = 15] (r5);
			
			\draw[->] (01111) to (r6);
			\draw[->] (r6) to (1111);
			\draw[->] (r6) to [bend right = 60] (0);
			\draw[->, dashed] (r6) to [bend right = 15] (010);
			\draw[->, dashed] (010) to [bend right = 15] (r6);
			
			\draw[->] (0) to (r7);
			\draw[->] (10) to (r7);
			\draw[->] (r7) to (010);
			\draw[->, dashed] (r7) to [bend right = 30] (1111);
			\draw[->, dashed] (1111) to [bend right = 30] (r7);
		\end{tikzpicture}
	\end{minipage}%
	\hfill%
	\begin{minipage}[c]{0.6\textwidth}
		\par\bigskip
		\par\bigskip
		\par\bigskip
		\par\bigskip
		\par\bigskip
		\par\bigskip
		\par\bigskip
		\par\bigskip
		\par\bigskip
		\par\bigskip
		\par\bigskip
		\par\bigskip
		\begin{equation*}
			\begin{split}
				r_1&: 01 + 11 + (100) \rightarrow 0111 + (100) \\
				r_2&: 11 + 00 + (11100) \rightarrow 1100 + (11100) \\
				r_3&: 0111 + 1100 (011) \rightarrow 01111100 + (011) \\
				r_4&: 01111100 + (1100) \rightarrow 011 + 11100 + (1100) \\
				r_5&: 01111100 + (1111) \rightarrow 01111 + 100 + (1111) \\
				r_6&: 01111 + (010) \rightarrow 0 + 1111 + (010) \\
				r_7&: 0 + 10 + (1111) \rightarrow 010 + (1111)
			\end{split}
		\end{equation*}
	\end{minipage}
	\caption{Example of an irreducible RAF from the BPM with multiple strong blocks.}
	\label{fig:BPMMultipleStrongBlocks}	
\end{figure}

\subsection{Compositions of irreducible RAFs}
Fig.~\ref{fig:irrRAFOverlap} and Fig.~\ref{fig:irrRAFCoreComp} show the overlap and 
composition in terms of F-reactions of the 50 arbitrarily chosen irreducible RAFs from 
two instances of the binary polymer model, i.e., n6\_split and n6\_forw, depicted in 
Fig.~\ref{fig:irrRAFCores} of the main text.

\begin{figure}[bth]
	\centering
	\begin{minipage}[c]{0.5\textwidth}
		\centering
		\includegraphics[width=\textwidth]{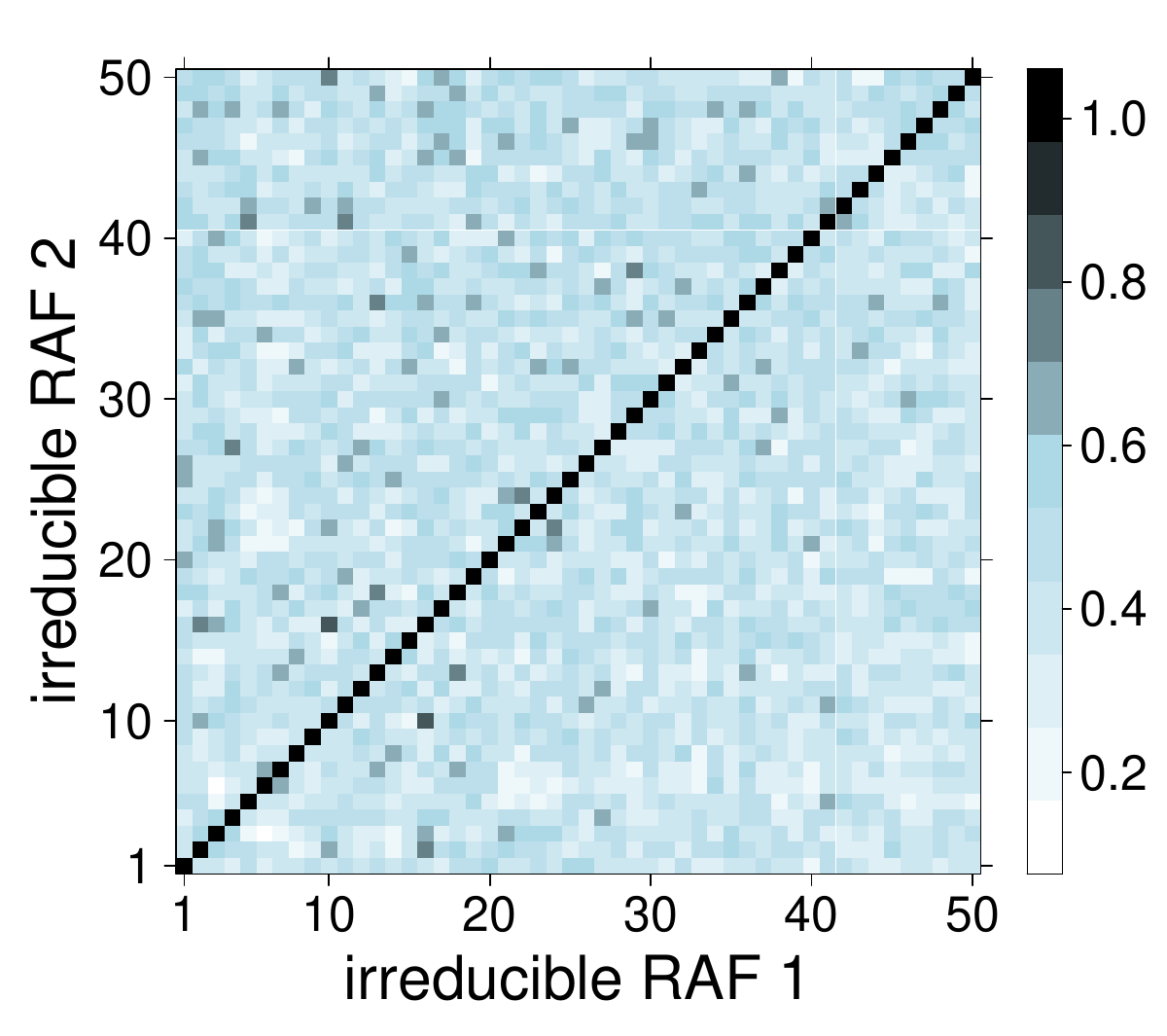}
	\end{minipage}%
	\hfill%
	\begin{minipage}[c]{0.5\textwidth}
		\includegraphics[width=\textwidth]{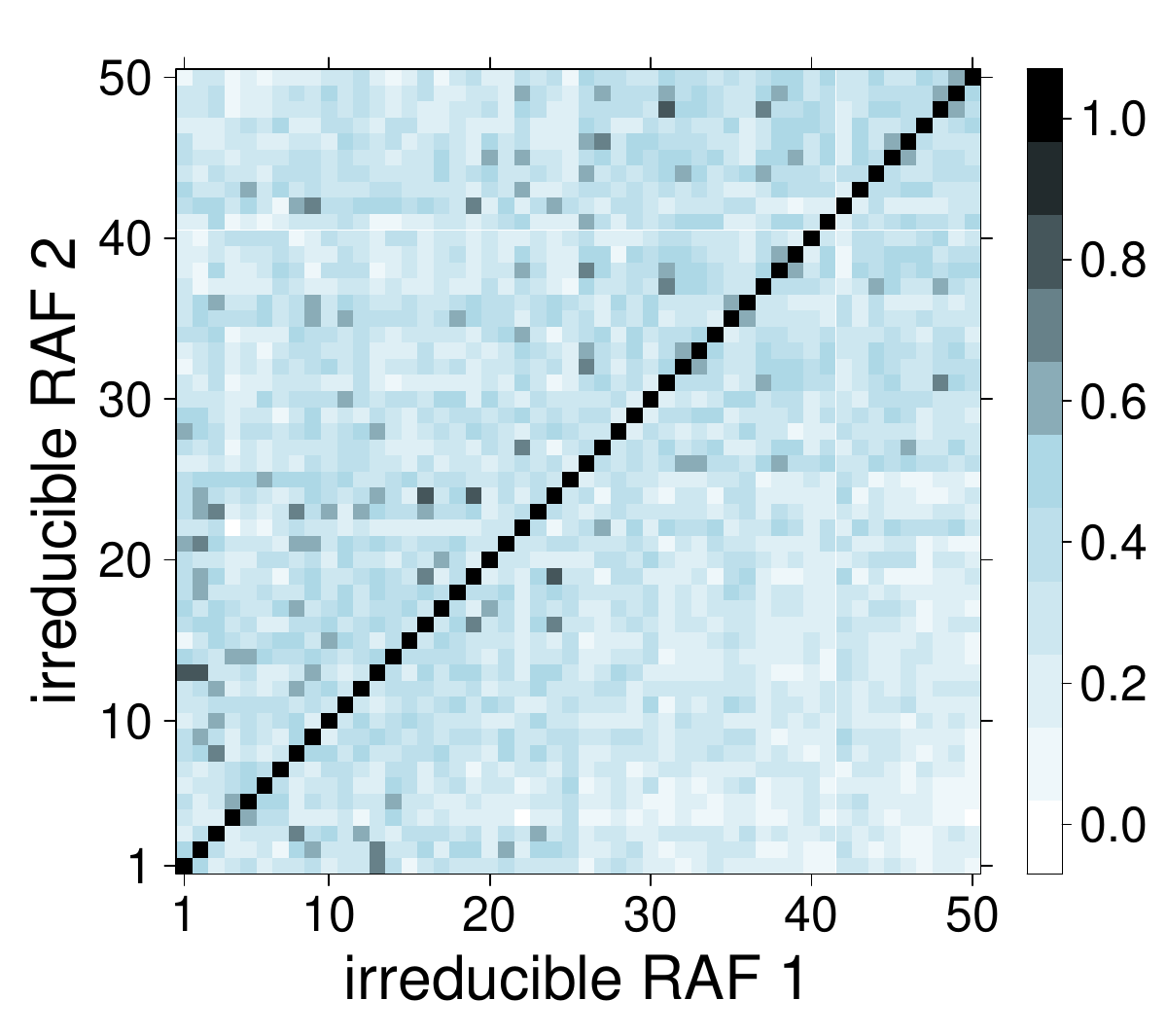}
	\end{minipage}
	\caption{Overlap w.r.t. F-reactions in the 50 irreducible RAFs from two 
	instances of the BPM with ligation/cleavage (\textbf{left}) or ligation 
	reactions only (\textbf{right}).}
	\label{fig:irrRAFOverlap}
\end{figure}

\begin{figure}
	\centering
	\begin{minipage}[c]{0.5\textwidth}
		\centering
		\includegraphics[width=\textwidth]{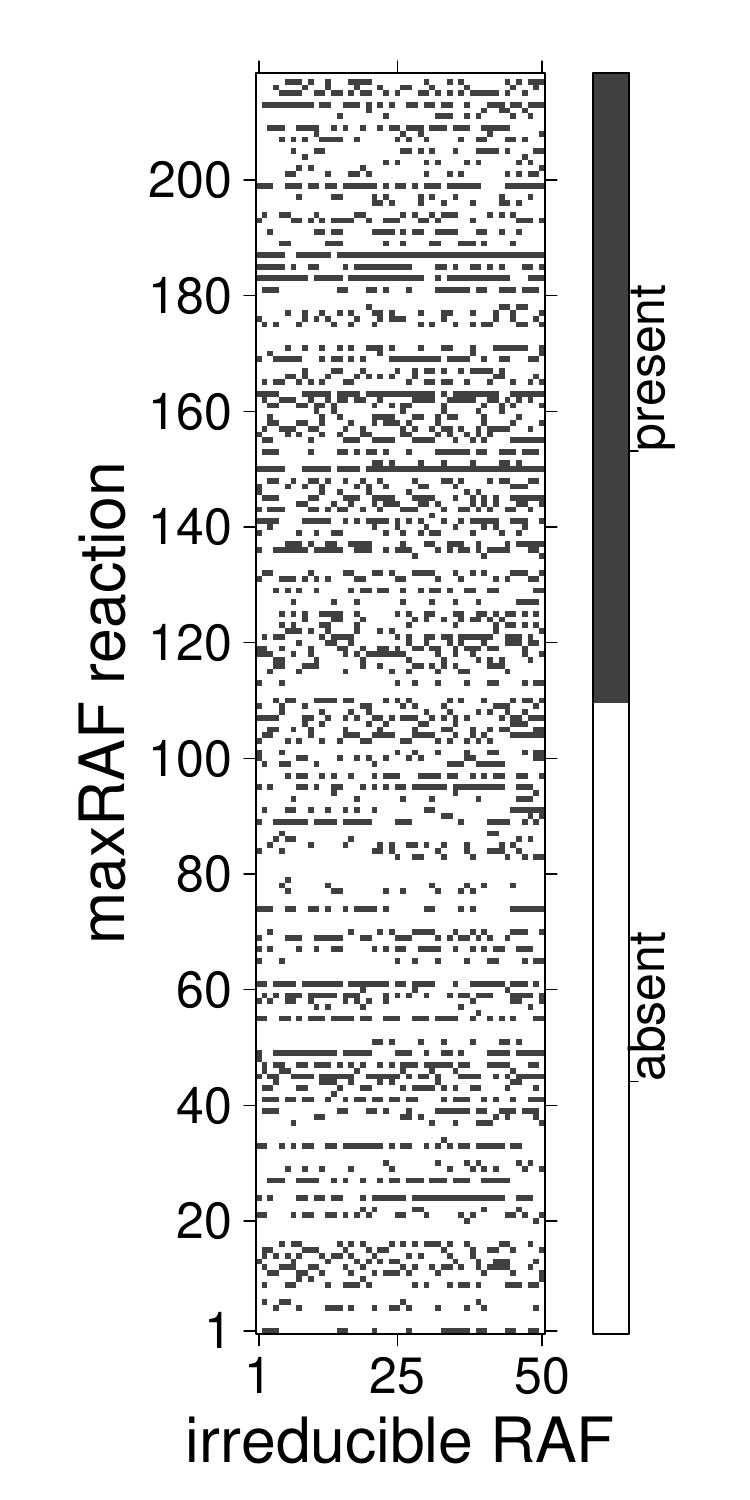}
	\end{minipage}%
	\hfill%
	\begin{minipage}[c]{0.5\textwidth}
		\includegraphics[width=\textwidth]{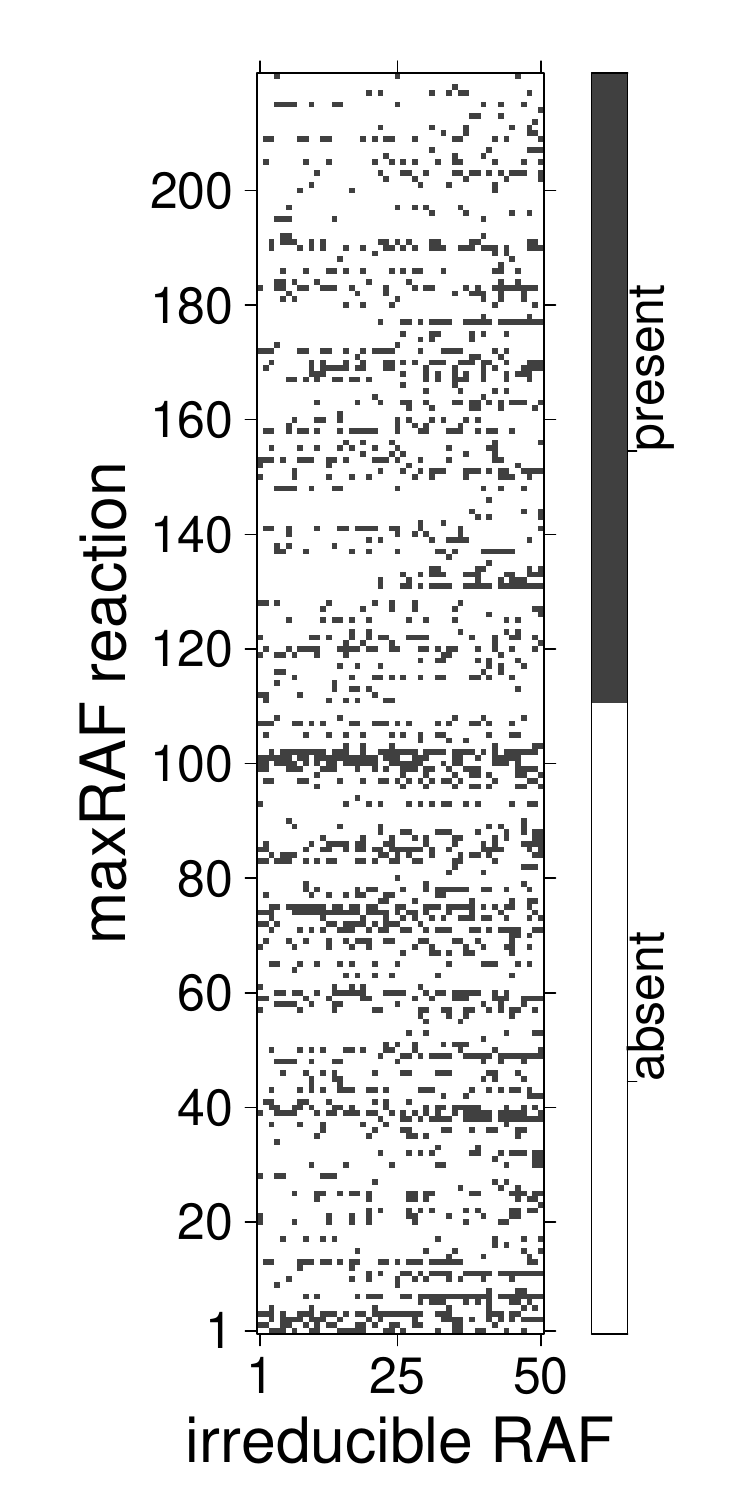}
	\end{minipage}
	\caption{Depiction of the F-reaction composition of 50 arbitrarily chosen 
	irreducible RAFs from two instances of the binary polymer model as depicted in 
	Fig.~\ref{fig:irrRAFCores} of the main text with maximal polymer size of $n=6$. See  
	Table~\ref{tab:SimulationDetails} for simulation details.}
	\label{fig:irrRAFCoreComp}
\end{figure}

\end{document}